\documentclass{scrartcl}
\usepackage{url}
\usepackage{textcomp}
\usepackage{pifont}
\usepackage{latexsym} 
\usepackage{amsfonts} 
\usepackage{amssymb}
\usepackage{amsmath}
\usepackage{amsthm}
\usepackage{stmaryrd}
\usepackage[english]{babel}
\usepackage{pgf}
\usepackage{stmaryrd}
\usepackage{acronym}
\usepackage{slashed} 
\usepackage{tikz}
\usepackage{mathtools}
\usepackage{color}
\usepackage{pgf}
\usepackage{tikz}
\usetikzlibrary{backgrounds}
\usetikzlibrary{fit}
\usetikzlibrary{matrix}
\usepackage{rta}

\newcommand{\email}[1]{\url{#1}}

\newtheorem{theorem}{Theorem}[section]
\newtheorem{corollary}[theorem]{Corollary}
\newtheorem{lemma}[theorem]{Lemma}

\newtheorem{definition}[theorem]{Definition}

\newtheorem{example}[theorem]{Example}

\begin{document}

\title{Technical Report: Complexity Analysis by Graph Rewriting Revisited%
\thanks{This research is supported by FWF (Austrian Science Fund) projects P20133.}%
\\\small{Technical Report}
}%
\author{Martin Avanzini\\
Institute of Computer Science\\
University of Innsbruck, Austria\\
\url{martin.avanzini@uibk.ac.at}
\and 
Georg Moser\\
Institute of Computer Science\\
University of Innsbruck, Austria\\
\url{georg.moser@uibk.ac.at}}
\maketitle

\begin{abstract}
\noindent
Recently, many techniques have been introduced
that allow the (automated) classification of
the runtime complexity
of term rewrite systems (TRSs for short).
In earlier work, the authors have shown
that for confluent TRSs, innermost polynomial runtime complexity
induces polytime computability of the functions defined.  
  
In this paper, we generalise the above result 
to full rewriting. Following our previous work, we exploit
graph rewriting. We give a new proof of the adequacy of 
graph rewriting for full rewriting that allows for
a precise control of the resources copied. In sum we
completely describe an implementation of rewriting
on a Turing machine (TM for short). We show that the runtime complexity of
the TRS and the runtime complexity of the TM is polynomially
related.
Our result strengthens the evidence that the complexity of
a rewrite system is truthfully represented through the length
of derivations. Moreover our result allows the classification of
non-deterministic polytime-computation based on runtime 
complexity analysis of rewrite systems. 
\end{abstract}

\newpage
\tableofcontents
\newpage

\section{Introduction}

Recently we see increased interest into studies of the (maximal) derivation
length of term rewrite system, compare for example~\cite{KW:2008,EWZ:2008,MSW:2008,KM:2009,MS09}.
We are interested in techniques to automatically
classify the complexity of term rewrite systems (TRS for short) and 
have introduced the \emph{polynomial path order} 
\POPSTAR and extensions of it, cf.~\cite{AM08,AM09}. 
\POPSTAR~is a restriction of the multiset path order~\cite{TeReSe} and whenever
compatibility of a TRS $\RS$ with \POPSTAR~can be shown then the runtime
complexity of $\RS$ is polynomially bounded. 
Here the \emph{runtime complexity} of a TRS measures the maximal number of rewrite steps as a function in the
size of the initial term, where the initial terms are restricted 
argument normalised terms (aka \emph{basic} terms).  

We have successfully implemented this technique.%
\footnote{Our implementation forms part of the 
\emph{Tyrolean Complexity Tool} (\TCT\ for short). For further information, 
see~\url{http://cl-informatik.uibk.ac.at/software/tct/}.}
As a consequence we can automatically verify for a given TRS $\RS$ that it admits at most 
polynomial runtime complexity. In this paper we study the
question, whether such results are restricted to \emph{runtime complexity} or can be applied
also for the (automated) classification of the intrinsic 
\emph{computational complexity} of the functions computed by the given
TRS $\RS$.
For motivation consider the TRS given in the
next example. It is not difficult to see that $\RSsat$ encodes
the function problem $\FSAT$ associated to the well-known 
satisfiability problem $\SAT$.
\begin{example}
\label{ex:1}
Consider the following TRS $\RSsat$:
\begin{alignat*}{4}
  1 \colon &\;&  \mif(\mtrue,t,e) & \to t &
  11 \colon &\;& \varepsilon = \varepsilon & \to \mtrue \\
  2 \colon &\;&  \mif(\mfalse,t,e) & \to e & 
  12 \colon &\;& \mO(x) = \mO(y) & \to x = y \\
  3 \colon &\;&  \choice(x \cons xs) & \to x &                                    
  13 \colon &\;& \mO(x) = \mZ(y) & \to \mfalse \\
  4 \colon &\;& \choice(x \cons xs) & \to \choice(xs) &                          
  14 \colon &\;& \mZ(x) = \mO(y) & \to \mfalse \\
  5 \colon &\;& \guess(\nil) & \to \nil &                          
  15 \colon &\;& \mZ(x) = \mZ(y) & \to x = y \\
  6 \colon &\;& \guess(c \cons cs) & \to \choice(c) \cons \guess(cs) &
  16 \colon &\;& \verify(\nil) & \to \mtrue \\
  7 \colon &\;& \member(x,\nil) & \to \mfalse &
  17 \colon &\;& \verify(l \cons ls) & \to \mif(\member(\neg l,ls), \mfalse, \verify(ls)) \\
  8 \colon &\;& \member(x,y \cons ys) & \to \mif(x = y, \mtrue, \member(x,ys)) &
  ~~ 18 \colon &\;& \issat'(c,a) & \to \mif(\verify(a),a,\unsat) \\
  9 \colon &\;& \neg \mO(x) & \to \mZ(x) & 
  19 \colon &\;& \issat(a) & \to \issat'(c,\guess(c)) \\
  10\colon &\;& \neg \mZ(x) & \to \mO(x)
\end{alignat*}
These rules are compatible with \POPSTAR\ and as a result
we conclude that the innermost runtime complexity induced
is polynomially.%
\footnote{To our best knowledge \TCT\ is currently the only complexity
tool that can provide a complexity certificate for the TRS $\RSsat$,
compare~\url{http://termcomp.uibk.ac.at}.}
\end{example}

$\FSAT$ is complete for the class of function problems over $\NP$ ($\FNP$ for short),
compare~\cite{Papa}.
This leads to the question, whether a characterisation 
of the runtime complexity of $\RSsat$ suffices to conclude that the function
computed by $\RSsat$ belongs to the class $\FNP$. 
The purpose of this paper is to provide a positive answer to this question.
More precisely, we establish the following results:
\begin{itemize}
\item We re-consider graph rewriting and provide a new proof of the adequacy of 
graph rewriting for full rewriting. This overcomes obvious inefficiencies
of rewriting, when it comes to the duplication of results.
\item We provide a precise analysis of the resources needed in
implementing graph rewriting on a Turing machine (TM for short).
\item Combining these results we obtain an efficient implementation of rewriting
on a TM. Based on this implementation our main result on the
correspondence between polynomial runtime complexity and polytime computability
follows.
\end{itemize}

Our result strengthens the evidence that the complexity of
a rewrite system is truthfully represented through the length
of derivations. Moreover our result allows the classification of
nondeterministic polytime-computation based on runtime 
complexity analysis of rewrite systems. 
This extends previous work (see~\cite{AM10}) that shows that 
for confluent TRSs, innermost polynomial runtime complexity
induces polytime computability of the functions defined. Moreover
it extends related work by Dal Lago and Martini (see~\cite{LM09,LM09b}) 
that studies the complexity of orthogonal TRSs, also applying
graph rewriting techniques.

\medskip
The paper is structured as follows. In Section~\ref{Preliminaries} we 
present basic notions, in Section~\ref{Grewrite} we (briefly) recall the central concepts of
graph rewriting. The adequacy theorem is provided in Section~\ref{Simulation} 
and in Section~\ref{Complexity} we show how rewriting can be implemented
efficiently. Finally we discuss our results in Section~\ref{Discussion},
where the above application to computational complexity is made precise.


\section{Preliminaries}\label{Preliminaries}

We assume familiarity with the basics of term rewriting, see~\cite{BN98,TeReSe}.
No familiarity with graph rewriting (see~\cite{TeReSe}) is assumed.
Let $R$ be a binary relation on a set $S$. We write $R^{+}$ for the transitive and
$R^{*}$ for the transitive and reflexive closure of $R$.  
An element $a \in S$ is \emph{$R$-minimal} if there exists no $b \in S$ 
such that $a \mathrel{R} b$.

Let $\VS$ denote a countably infinite set of variables and $\FS$ a 
signature, containing at least one constant. 
The set of terms over $\FS$ and $\VS$ is denoted as $\TERMS$.
The \emph{size} $\size{t}$ of a term $t$ is defined as usual.
A \emph{term rewrite system} (\emph{TRS} for short) $\RS$ over
$\TERMS$ is a \emph{finite} set of rewrite
rules $l \to r$, such that $l \notin \VS$ and $\Var(l) \supseteq \Var(r)$.
We write $\rew$ for the induced rewrite
relation. 
The set of defined function symbols is denoted as $\DS$, while 
the constructor symbols are collected in $\CS$, 
clearly $\FS = \DS \cup \CS$.
We use $\NF(\RS)$ to denote the set of normal forms of $\RS$ and
 write $s \rsn t$ if $s \rss t$ and $t \in \NF(\RS)$.
We define the set of \emph{values} $\Val \defsym \TA(\CS,\VS)$, 
and we define $\TB \defsym \{f(\seq{v}) \mid f \in \DS 
\text{ and } v_i \in \Val\}$ as the set of \emph{basic terms}.
Let $\hole$ be a fresh constant. Terms over $\FS \cup \set{\hole}$ and $\VS$
are called \emph{contexts}. 
The empty context is denoted as $\hole$. 
For a context $C$ with $n$ holes, 
we write $C[\seq{t}]$ for the term obtained by replacing the holes 
from left to right in $C$ with the terms $\seq{t}$.

A TRS is called \emph{confluent} if for all $s, t_1, t_2 \in \TERMS$
with $s \rss t_1$ and $s \rss t_2$ there exists a term $t_3$ such that
$t_1 \rss t_3$ and $t_2 \rss t_3$.
The \emph{derivation length} of a terminating term $s$ with respect to
$\to$ is defined as ${\dl(s,\to)} \defsym \max\{ n \mid \exists t. \; s \to^n t \}$,
where $\to^n$ denotes the $n$-fold application of $\to$. 
The \emph{runtime complexity function} $\rcR$ with respect to
a TRS $\RS$ is defined as 
$\rcR(n) \defsym \max\{ \dl(t, \rew) \mid \text{$t \in \TB$ and $\size{t} \leqslant n$} \}$.


\section{Term Graph Rewriting}\label{Grewrite}

In the sequel we introduce the central concepts of~\emph{term graph rewriting}
or \emph{graph rewriting} for short. 
We closely follow the presentation of \cite{AM10}, 
for further motivation of the presented notions we kindly 
refer the reader to \cite{AM10}.
Let $\RS$ be a TRS over a signature $\FS$. We keep $\RS$ and $\FS$ 
fixed for the remainder of this paper. 

A \emph{graph} $G=(\nodes[G],\suc[G],\lab[G])$ over the set $\LS$ of \emph{labels} is a structure
such that $\nodes[G]$ is a finite set, the \emph{nodes} or \emph{vertexes},
$\suc \colon \nodes[G] \to \nodes[G]^{\ast}$ is a mapping that associates
a node $u$ with an (ordered) sequence of nodes, called the \emph{successors} of $u$.
Note that the sequence of successors of $u$ may be empty: $\suc[G](u) = []$.
Finally $\lab[G] \colon \nodes[G] \to \LS$ is a mapping that associates each
node $u$ with its \emph{label} $\lab[G](u)$.
Typically the set of labels $\LS$ is clear from context and not explicitly mentioned. 
In the following, nodes are denoted by $u,v, \dots$ possibly followed by subscripts. 
We drop the reference to the graph $G$
from $\nodes[G]$, $\suc[G]$, and $\lab[G]$, i.e., 
we write $G = (\Nodes,\suc,\Lab)$ if no confusion can arise from this.
Further, we also write $u \in G$ instead of $u \in \Nodes$.

Let $G=(\Nodes,\suc,\Lab)$ be a graph and let $u \in G$. Consider
$\suc(u) = [\seq[k]{u}]$. We call $u_i$ ($1 \leqslant i \leqslant k$) the \emph{$i$-th successor} of $u$
(denoted as $u \suci[G]{i} u_i$). 
If $u \suci[G]{i} v$ for some $i$, then we simply write
$u \reach[G] v$. A node $v$ is called \emph{reachable} from $u$ if $u \reachtr[G] v$,
where $\reachtr[G]$ denotes the reflexive and transitive closure of $\reach[G]$.
We write $\reachtir[G]$ for $\reach[G] \cdot \reachtr[G]$.
A graph $G$ is \emph{acyclic} if $u \reachtir[G] v$ implies $u \not= v$ and
$G$ is \emph{rooted} if there exists a unique node $u$ such that every other node 
in $G$ is reachable from $u$. The node $u$ is called the \emph{root} $\grt(G)$ of $G$. 
The \emph{size} of $G$, i.e., the number of nodes, is denoted as $\size{G}$.
The \emph{depth} of $S$, i.e., the length of 
the longest path in $S$, is denoted as $\depth(S)$.
We write $\subgraphAt{G}{u}$ for the subgraph of $G$ reachable from $u$.

Let $G$ and $H$ be two term graphs, possibly sharing nodes. 
We say that $G$ and $H$ are \emph{properly sharing} if 
$u \in G \cap H$ implies $\lab[G](u) = \lab[H](u)$ and $\suc[G](u) = \suc[H](u)$.
If $G$ and $H$ are properly sharing, we write $G \graphUnion H$ for their union.
\begin{definition}
  \label{d:termgraph}
  A \emph{term graph} (with respect to $\FS$ and $\VS$) is an \emph{acyclic} and \emph{rooted} graph 
  $S = (\Nodes,\suc,\Lab)$ over labels $\FS \cup \VS$. 
  Let $u \in S$ and suppose $\Lab(u) = f \in \FS$ such that
  $f$ is $k$-ary. Then $\suc(u) = [\seq[k]{u}]$. 
  On the other hand, if $\Lab(u) \in \VS$ then $\suc(u) = []$. 
  We demand that any variable node is \emph{shared}. That is,
  for $u \in S$ with $\Lab(u) \in \VS$, if $\Lab(u) = \Lab(v)$ for some $v \in \Nodes$ 
  then $u=v$.
\end{definition}
Below $S, T, \dots$ and $L,R$, possibly followed by subscripts, always denote term graphs.
We write $\GRAPHS$ for the set of all term graphs with respect to $\FS$ and $\VS$. 
Abusing notation from rewriting we set $\Var(S) \defsym \set{u \mid u \in S, \Lab(u) \in \VS}$,  
the set of \emph{variable nodes} in $S$. 
We define the term $\trepr(S)$ \emph{represented} by $S$ as follows:
$\trepr(S) \defsym x$ if $\Lab(\grt(S)) = x \in \VS$ and
$\trepr(S) \defsym f(\trepr(\subgraphAt{S}{u_1}), \dots, \trepr(\subgraphAt{S}{u_k}))$ for $\Lab(\grt(S)) = f \in \FS$ 
and $\suc(\grt(S)) = [\seq[k]{u}]$.

We adapt the notion of \emph{positions} in terms to positions in graphs in the obvious way.
Positions are denoted as $p,q, \dots$, possibly followed by subscripts. 
For positions $p$ and $q$ we write $p q$ for their concatenation. 
We write $p \leqslant q$ if $p$ is a prefix of $q$, i.e., $q = p p'$ for some 
position $p'$. 
The size $\size{p}$ of position $p$ is defined as its length.
Let $u \in S$ be a node.
The set of \emph{positions} $\Pos[S](u)$ of $u$ is 
defined as $\Pos[S](u) \defsym \set{\varepsilon}$ if $u = \grt(S)$ and
$\Pos[S](u) \defsym \set{i_1 \cdots i_k  \mid \grt(S) \suci[S]{i_1} \cdots \suci[S]{i_{k}} u}$
otherwise.
The set of all positions in $S$ is $\Pos[S] \defsym \bigcup_{u \in S} \Pos[S](u)$.
Note that $\Pos[S]$ coincides with the set of positions of $\trepr(S)$.
If $p \in \Pos[S](u)$ we say that $u$ \emph{corresponds} to $p$. 
In this case we also write $\subgraphAt{S}{p}$ for the subgraph $\subgraphAt{S}{u}$.
This is well defined since exactly one node corresponds to a position $p$.
One verifies $\trepr(\subgraphAt{S}{p}) = \subtermAt{\trepr(S)}{p}$ for all $p \in \Pos[S]$.
We say that $u$ is (strictly) \emph{above} a position $p$ if $u$ corresponds to a
position $q$ with $q \leqslant p$ ($q < p$).
Conversely, the node $u$ is \emph{below} $p$ if $u$ corresponds to $q$ with $p \leqslant q$.

By exploiting different degrees of \emph{sharing}, a term $t$ can 
often be represented by more than one term graph.
Let $S$ be a term graph and let $u \in S$ be a node.
We say that $u$ is \emph{shared} 
if the set of positions $\Pos[S](u)$ is not singleton. 
Note that in this case, the node $u$ represents more than one subterm of $\trepr(S)$.
If $\Pos[S](u)$ is singleton, then $u$ is \emph{unshared}.
The node $u$ is \emph{minimally shared} if it is a variable node or unshared otherwise
(recall that variable nodes are always shared).
We say $u$ is \emph{maximally shared} if
$\trepr(\subgraphAt{S}{u}) = \trepr(\subgraphAt{S}{v})$ 
implies $u = v$.
The term graph $S$ is called \emph{minimally sharing} (\emph{maximally sharing})
if all nodes $u \in S$ are minimally shared (maximally shared).
Let $s$ be a term.
We collect all minimally sharing term graphs representing $s$ in the set $\Tree(s)$.
Maximally sharing term graphs representing $s$ are collected in $\Shared(s)$.

We now introduce a notion for replacing a subgraph $\subgraphAt{S}{u}$ 
of $S$ by a graph $H$.
\begin{definition}
Let $S$ be a term graph and let $u,v \in S$ be two nodes.
Then $\replaceNode{S}{u}{v}$ denotes the 
\emph{redirection} of node $u$ to $v$:
set $r(u) \defsym v$ and $r(w) \defsym w$ for all $w \in S \setminus \set{u}$.
Set $\Nodes' \defsym (\nodes[S] \cup \set{v}) \setminus \set{u}$
and for all $w \in \Nodes'$, $\suc'(w) \defsym r^\ast(\suc[S](w))$
where $r^\ast$ is the extension of $r$ to sequences.
Finally, set $\replaceNode{S}{u}{v} \defsym (\Nodes',\suc',\lab[S])$.

Let $H$ be a rooted graph over $\FS \cup \VS$.
We define $\replaceAt{S}{u}{H} \defsym \subgraphAt{(\replaceNode{S}{u}{\grt(H)} \graphUnion H)}{v}$
where $v = \grt(H)$ if $u = \grt(S)$ and $v = \grt(S)$ otherwise.
Note that $\replaceAt{S}{u}{H}$ is again a term graph if
$u \not\in H$ and $H$ acyclic.
\end{definition}

The following notion of \emph{term graph morphism} plays the r\`ole of substitutions.
\begin{definition}
  \label{d:morphismus}
  Let $L$ and $T$ be two term graphs. 
  A \emph{morphism} from $L$ to $T$ (denoted $m \colon L \to T$) is a function 
  $m \colon \nodes[L] \to \nodes[T]$ such that 
  $m(\grt(L)) = \grt(T)$, and for all $u \in L$ with $\lab[L](u) \in \FS$, 
  (i) $\lab[L](u) = \lab[T](m(u))$ and (ii) $m^\ast(\suc[L](u)) = \suc[T](m(u))$.
\end{definition}
The next lemma follows essentially from assertion (ii) of Definition \ref{d:morphismus}.
\begin{lemma} 
\label{l:morph:subgraph}
If $m \colon L \to S$ then for any $u \in S$ we have $m \colon \subgraphAt{L}{u} \to \subgraphAt{S}{m(u)}$.
\end{lemma}
\begin{proof}
  By a straight forward inductive argument.
\end{proof}
Let $m \colon L \to S$ be a morphism from $L$ to $S$. 
The \emph{induced substitution} $\sigma_m \colon \Var(L) \to \TA$ is defined 
as $\sigma_m(x) \defsym \trepr(\subgraphAt{S}{m(u)})$ for any 
$u \in S$ such that $\Lab(u) = x \in \VS$.
As an easy consequence of Lemma~\ref{l:morph:subgraph} we obtain the following.
\begin{lemma}\label{l:subst:l}
  Let $L$ and $S$ be term graphs, and
  suppose $m \colon L \to S$ for some morphism $m$.
  Let $\sigma_m$ be the substitution induced by $m$.
  Then $\trepr(L)\sigma_m = \trepr(S)$.
\end{lemma}
\begin{proof}
The lemma has also been shown in \cite[Lemma 14]{AM10}. For completeness 
we restate the proof.

We prove that for each node $u \in L$, 
$\trepr(\subgraphAt{L}{u})\sigma_m = \trepr(\subgraphAt{S}{m(u)})$
by induction on $l \defsym \trepr(\subgraphAt{L}{u})$.
To conclude $\trepr(L)\sigma_m = \trepr(S)$,
it suffices to observe that by definition of $m$: $m(\grt(L)) = \grt(S)$. 

If $l \in \VS$, then by definition $\subgraphAt{L}{u}$ consists of a single (variable-)node
and $l\sigma_m = \trepr(\subgraphAt{S}{m(u)})$ follows by the
definition of the induced substitution $\sigma_m$.
If $l = f(l_1,\dots,l_k)$, then we have $\suc[L](u) = [u_1,\dots,u_k]$ for some
$\seq[k]{u} \in L$. 
As $m \colon L \to S$ holds, Lemma~\ref{l:morph:subgraph} yields
$m \colon \subgraphAt{L}{u_i} \to \subgraphAt{S}{m(u_i)}$ for all
$i = 1, \dots, k$. And induction hypothesis becomes
applicable so that $l_i\sigma_m = \trepr(\subgraphAt{S}{m(u_i)})$.
Thus 
\begin{equation*}
  l\sigma_m = f(l_1\sigma_m,\dots,l_k\sigma_m) = 
  f(\trepr(\subgraphAt{S}{m(u_1)}), \dots, \trepr(\subgraphAt{S}{m(u_k)})) \tpkt
\end{equation*}
By definition of $m$, $\lab[S](m(u)) = \lab[L](u) = f$ and 
$\suc[S](m(u)) = m^{*}(\suc[L](u)) = [m(u_1), \dots, m(u_k)]$. Hence
$f(\trepr(\subgraphAt{S}{m(u_1)}), \dots, \trepr(\subgraphAt{S}{m(u_k)})) 
= \trepr(\subgraphAt{S}{m(u)})$.
\end{proof}

We write $S \geqm T$ (or $S \geqslant T$ for short)
if $m \colon S \to T$ is a morphism such that for \emph{all}
$u \in \nodes[S]$, Property (i) and Property (ii) in Definition \ref{d:morphismus} are fulfilled.
For this case, $S$ and $T$ represent the same term.
We write $S \gem T$ (or $S > T$ for short)
when the graph morphism $m$ is additionally \emph{non-injective}.
If both $S \geqm[] T$ and $T \geqm[] S$ holds then $S$ and $T$ are \emph{isomorphic}, 
in notation $S \isomorphic T$.
Recall that $\size{S}$ denotes the number of nodes in $S$.

\begin{lemma}\label{l:morph:termeq}
  For all term graph $S$ and $T$, $S \geqm T$ implies $\trepr(S) = \trepr(T)$ and $\size{S} \geqslant \size{T}$.
  If further $S \gem T$ holds then $\size{S} > \size{T}$.
\end{lemma}
\begin{proof}
Suppose $S \geqm T$
We first prove $\trepr(S) = \trepr(T)$. 
We prove the lemma by induction on $S$. 
For the base case, suppose $S$ consists of a single
node $u$ such that $\lab[S](u) = x \in \VS$. As the node $u$ is the
root of $S$, definition of $m$ yields that $T$ consists of a
single (variable)-node labeled with $x$. (Note that $m$ is
a surjective morphism.) Thus the result follows trivially.
For the inductive step, suppose $\lab[S] = f \in \FS$ and
$\suc[S] = [u_1,\dots,u_k]$. From $S \geqm T$ we see that
$\lab[T] = f$ and $\suc[T] = [m(u_1),\dots,m(u_k)]$. 
As a consequence of Lemma~\ref{l:morph:subgraph}, 
$\subgraphAt{S}{u_i} \geqm \subgraphAt{T}{m(u_i)}$ and thus 
by induction hypothesis $\trepr(\subgraphAt{S}{u_i}) = \trepr(\subgraphAt{T}{m(u_i)})$
as desired.

Now observe that $m$ is by definition surjective. Thus 
$\size{S} \geqslant \size{T}$ trivially follows. 
Further, it is not hard to see that if $m$ is additionally non-injective, 
i.e., $S \gem T$, then clearly $\size{S} > \size{T}$.
\end{proof}

Let $L$ and $R$ be two properly sharing term graphs.
Suppose $\grt(L) \not \in \Var(L)$, $\Var(R) \subseteq \Var(L)$ and $\grt(L) \not \in R$. 
Then the graph $L \graphUnion R$ is called a \emph{graph rewrite rule} (\emph{rule} for short), 
denoted by $L \to R$. The graph $L$, $R$ denotes the left-hand, right-hand side of $L \to R$
respectively.
A \emph{graph rewrite system} (\emph{GRS} for short) $\GS$ is a set of graph rewrite rules.

Let $\GS$ be a GRS, let $S \in \GRAPHS$ and let $L \to R$
be a rule. 
A rule $L' \to R'$ is called a \emph{renaming} of $L \to R$ 
with respect to $S$ if $(L' \to R') \isomorphic (L \to R)$ 
and $\nodes[S] \cap \nodes[{L' \to R'}] = \varnothing$.
Let ${L' \to R'}$ be a renaming of a rule $(L \to R) \in \GS$ for $S$, 
and let $u \in S$ be a node.
We say $S$ \emph{rewrites} to $T$ at \emph{redex} $u$ with rule $L \to R$,
denoted as $S \REW[\GS,u, L \to  R] T$, 
if there exists a morphism $m \colon L' \to \subgraphAt{S}{u}$
and $T = \replaceAt{S}{u}{m(R')}$.
Here $m(R')$ denotes the structure obtained by replacing in $R'$
every node $v \in \dom(m)$ by $m(v) \in S$, where 
the labels of $m(v) \in m(R')$ are the labels of $m(v) \in S$.
We also write $S \REW[\GS,p, L \to  R] T$ if 
$S \REW[\GS,u, L \to  R] T$ for the position $p$ corresponding to $u$ in $S$.
We set $S \REW[\GS] T$ if $S \REW[\GS,u, L \to  R] T$ for some $u \in S$ and $(L \to R) \in \GS$.
The relation $\REW[\GS]$ is called the \emph{graph rewrite relation}
induced by $\GS$. 
Again abusing notation, we denote the set of normal-forms with respect to $\REW$ as
$\NF(\GS)$.


\section{Adequacy of Graph Rewriting for Term Rewriting}
\label{Simulation}

In earlier work \cite{AM10} we have shown that graph rewriting is adequate for 
innermost rewriting without further restrictions on the studied TRS $\RS$. 
In this section we generalise this result to full rewriting.
The here presented adequacy theorem (see Theorem~\ref{t:adequacy}) is not essentially new. 
Related results can be found in the extensive literature, 
see for example \cite{TeReSe}. 
In particular, in~\cite{DetlefPlump:01} the adequacy theorem is stated for 
full rewriting and unrestricted TRSs.
In this work, we take a fresh look from a complexity related point of view.
%
We give a new proof of the adequacy of 
graph rewriting for full rewriting that allows for
a precise control of the resources copied. This is essential
for the accurate characterisation of the 
implementation of graph rewriting given in Section~\ref{Complexity}.

\begin{definition}
The \emph{simulating graph rewrite system $\GS(\RS)$} of $\RS$ 
contains for each rule ${(l \to r)} \in \RS$ some rule 
$L \to R$ such that $L \in \Tree(l)$, $R \in \Tree(r)$ and
$\nodes[L] \cap \nodes[R] = \Var(R)$. 
\end{definition}

The next two Lemmas establish soundness in the sense that
derivations with respect to $\GS(\RS)$ correspond to $\RS$-derivations.
\begin{lemma}\label{l:subst:r}
  Let $S$ be a term graph and 
  let $L \to R$ be renaming of a graph rewrite rule for $S$, i.e.,
  $S \cap R = \varnothing$.
  Suppose $m \colon L \to S$ for some morphism $m$ and let 
  $\sigma_m$ be the substitution induced by $m$.
  Then $\trepr(R)\sigma_m = \trepr(T)$ where $T \defsym \subgraphAt{(m(R) \cup S)}{\grt(m(R))}$.
\end{lemma}
\begin{proof}
We prove the more general statement that 
for each $u \in R$, 
$\trepr(\subgraphAt{R}{u})\sigma_m = \trepr(\subgraphAt{T}{m(u)})$, 
c.f. also \cite[Lemma 15]{AM10}.
First suppose $u \in R \cap L$. Then $\subgraphAt{R}{u} = \subgraphAt{L}{u}$
as $L$, $R$ are properly shared. 
Employing Lemma~\ref{l:subst:l}, we have
$\trepr(\subgraphAt{R}{u})\sigma_m = \trepr(\subgraphAt{L}{u})\sigma_m
= \trepr(\subgraphAt{S}{m(u)})$. From this the assertion follows.

Thus suppose $u \in R \setminus L$. This subcase we prove
by induction on $r \defsym \trepr(\subgraphAt{R}{u})$. The base
case $r \in \VS$ is trivial, as variables are shared in $L \to R$. 
For the inductive step, let $r = f(r_1,\dots,r_k)$
with $\suc[R](u) = [\seq[k]{u}]$. 
We identify $m$ with the extension of $m$ to all nodes in $R$.
The induction hypothesis yields $r_i\sigma_m = \trepr(\subgraphAt{T}{m(u_i)})$
for $i = 1,\dots,k$. By definition of $m(R)$: $m(u) = u \in m(R) \subseteq T$.
Hence $\suc[T](m(u)) = \suc[m(R)](u) = 
{m}^{\ast}(\suc[R](u)) = [m(u_1), \dots,m(u_k)]$.
Moreover $\lab[T](u) = \lab[m(R)](u) = f$ by definition.
We conclude 
$r\sigma_m =  f(r_1\sigma_m,\dots,r_k\sigma_m) =
f(\trepr(\subgraphAt{T}{m(u_1)}),\dots,\trepr(\subgraphAt{T}{m(u_k)})) 
= \trepr(\subgraphAt{T}{m(u)})$.
\end{proof}

In Section~\ref{Preliminaries} we introduced $\hole$ as designation of the
empty context. Below we write $\hole$ for the unique 
(up-to isomorphism) graph representing the constant $\hole$.
\begin{lemma}\label{l:context}
Let $S$ and $T$ be two be properly sharing term graphs, let $u \in S \setminus T$
and $C = \trepr(\replaceAt{S}{u}{\hole})$.
Then $\trepr(\replaceAt{S}{u}{T}) = C[\trepr(T),\dots,\trepr(T)]$.
\end{lemma}
\begin{proof}
We proceed by induction on the size of $S$, compare also \cite[Lemma 16]{AM10}.
In the
base case $S$ consists of a single node $u$. Hence the context
$C$ is empty and the lemma follows trivially. 

For the induction step we can assume without loss of
generality that $u \not= \grt(S)$.  
We assume $\lab[S](\grt(S)) = f \in \FS$ and 
$\suc[S](\grt(S)) = [v_1,\dots,v_k]$. 
For all $i$ ($1 \leqslant i \leqslant k$) such that
$v_i = u$ set $C_i = \hole$ and for all
$i$ such that $v_i \not= u$ but 
$\subgraphAt{(\replaceAt{S}{u}{T})}{v_i} = \replaceAt{(\subgraphAt{S}{v_i})}{u}{T}$
we set $C_i = \trepr(\replaceAt{(\subgraphAt{S}{v_i})}{u}{\hole})$.
In the latter sub-case induction hypothesis is applicable
to conclude
\begin{equation*}
  \trepr\bigl(\subgraphAt{\bigl(\replaceAt{S}{u}{T}\bigr)}{v_i} \bigr) =
  C_i[\trepr(T), \dots, \trepr(T)] \tpkt
\end{equation*}
Finally we set $C \defsym f(\seq[k]{C})$ and obtain $C = \trepr(\replaceAt{S}{u}{\hole})$.
In sum we have
$ \trepr(\replaceAt{S}{u}{T}) = f\bigl( \subgraphAt{\bigl(\replaceAt{S}{u}{T}\bigr)}{v_1}, \dots, 
\subgraphAt{\bigl(\replaceAt{S}{u}{T}\bigr)}{v_k}\bigr) = C[\overline{\trepr(T)}]$, 
where $\overline{\trepr(T)}$ denotes the sequences of terms $\trepr(T)$. 
\end{proof}

For non-left-linear TRSs $\RS$, $\REW[\GS(\RS)]$ does 
not suffice to mimic $\rew[\RS]$. This is clarified in the following example.
\begin{example}\label{ex:problem:1}
  Consider the TRS $\RSf \defsym \set{\m{f}(x) \to \m{eq}(x,\m{a});~\m{eq}(x,x) \to \top}$.
  Then $\RSf$ admits the derivation
  $$
  \m{f}(\m{a}) \rew[\RSf] \m{eq}(\m{a},\m{a}) \rew[\RSf] \top
  $$
  but $\GS(\RSf)$ cannot completely simulate the above sequence:
  \begin{center}
  \begin{tikzpicture}[node distance=7mm]
    \node (A)                             {$\m{f}$};
    \node (A1) [below of=A]               {$\m{a}$};
    \node (B) at (A) [xshift=2.2cm]       {$\m{eq}$};
    \node (B1) [below of=B, xshift=-3mm]  {$\m{a}$};
    \node (B2) [below of=B, xshift=3mm]   {$\m{a}$};

    \draw (A) -- (A1);
    \draw (B) -- (B1);
    \draw (B) -- (B2);
    \node at (A) [xshift=10mm, yshift=-3.5mm] {${\REW[\GS(\RSf)]}$};
    \node at (B) [xshift=16mm, yshift=-3.5mm] {${\in \NF(\GS(\RSf))}$};
  \end{tikzpicture}
\end{center}
 Let $L \to R$ be the rule in $\GS(\RSf)$ corresponding to $\m{eq}(x,x) \to \top$, 
 and let $S$, $\trepr(S) = \m{eq}(\m{a},\m{a})$, be the second graph in the above sequence.
 Then $L \to R$ is inapplicable as we cannot simultaneously map
 the unique variable node in $L$ to both leaves in $S$ via a graph morphism.
 Note that the situation can be repaired by sharing the two arguments in $S$.
\end{example}
For maximally sharing graphs $S$
we can prove that redexes of $\RS$ and (positions corresponding to) 
redexes of $\SGRS(\RS)$ coincide. This is a consequence of the following Lemma.
\begin{lemma}\label{l:match:lhs}
  Let $l$ be a term and $s = l\sigma$ for some substitution $\sigma$.
  If $L \in \Tree(l)$ and $S \in \Shared(s)$, then there exists a morphism 
  $m \colon L \to S$. Further, $\sigma(x) = \sigma_m(x)$ 
  for the induces substitution $\sigma_m$ and all variables $x \in \Var(l)$.
\end{lemma}
\begin{proof}
We prove the lemma by induction on $l$. It suffices
to consider the induction step. Let 
$l = f(\seq[k]{l})$ and $s = f(l_1\sigma,\dots,l_k\sigma)$. 
Suppose $\suc[L](\grt(L)) = [\seq[k]{u}]$ and $\suc[S](\grt(S)) = [\seq[k]{v}]$.
By induction hypothesis there exist morphisms 
$m_i \colon \subgraphAt{L}{u_i} \to \subgraphAt{S}{v_i}$ ($1 \leqslant i \leqslant k$) 
of the required form.
Define $m \colon \nodes[L] \to \nodes[S]$ as follows.
Set $m(\grt(L)) = \grt(S)$ and for 
$w \not = \grt(L)$ define $m(w) = m_i(w)$ if $w \in \dom(m_i)$.
We claim $w \in (\dom(m_i) \cap \dom(m_j))$ implies $m_i(w) = m_j(w)$.
For this, suppose $w \in (\dom(m_i) \cap \dom(m_j))$.
Since $L \in \Tree(l)$, only variable nodes are shared, hence
$w$ needs to be a variable node, say $\lab[L](w) = x \in \VS$.
Then 
$$\trepr(\subgraphAt{S}{m_i(w)}) = \sigma_{m_i}(x) = \sigma(x) = \sigma_{m_j}(x) = \trepr(\subgraphAt{S}{m_j(w)})$$
by induction hypothesis.
As $S \in \Shared(s)$ is maximally shared, $m_i(w) = m_j(w)$ follows.
We conclude $m$ is a well-defined morphism, further $m \colon L \to S$.
\end{proof}

A second problem is introduced by non-eager evaluation. Consider the following.
\begin{example}\label{ex:problem:2}
  Let $\RSg \defsym \set{\m{dup}(x) \to \m{c}(x,x);~\m{a} \to \m{b}}$.
  Then $\RSg$ admits the derivation
  $$
  \m{dup}(\m{a}) \rew[\RSg] \m{c}(\m{a},\m{a}) \rew[\RSg] \m{c}(\m{b},\m{a})
  $$
  but applying the corresponding rules in $\GS(\RSg)$ yields:
  \begin{center}
    \begin{tikzpicture}[node distance=7mm]
      \node (A)                             {$\m{dup}$};
      \node (A1) [below of=A]               {$\m{a}$};
      \node (B) at (A) [xshift=2.4cm]       {$\m{c}$};
      \node (B1) [below of=B]  {$\m{a}$};
      \node (C) at (B) [xshift=2.4cm]       {$\m{c}$};
      \node (C1) [below of=C]  {$\m{b}$};

      \draw (A) -- (A1);
      \path (B) edge [bend left] (B1);
      \path (B) edge [bend right] (B1);
      \path (C) edge [bend left] (C1);
      \path (C) edge [bend right] (C1);

      \node at (A) [xshift=12mm, yshift=-3.5mm] {${\REW[\GS(\RSg)]}$};
      \node at (B) [xshift=12mm, yshift=-3.5mm] {${\REW[\GS(\RSg)]}$};
    \end{tikzpicture}
  \end{center}
  Application of the first rule produces a shared redex. 
  Consequently the second step amounts to a parallel step in $\RSg$. 
\end{example}

To prove adequacy of graph rewriting for term rewriting and unrestricted TRSs, 
we follow the standard approach \cite{TeReSe,DetlefPlump:01}
where \emph{folding} (also called \emph{collapsing}) and \emph{unfolding} (also referred to as \emph{copying}) 
is directly incorporated in the graph rewrite relation.
Unlike in the cited literature, we 
employ a very restrictive form of folding and unfolding.
To this extend, we define for positions $p$ 
relations $\flds{p}$ and $\ufldp{p}$ on term graphs. 
Both relations preserve term structure.
However, when $S \flds{p} T$ holds then the
subgraph $\subgraphAt{T}{p}$ admits strictly more sharing than $\subgraphAt{S}{p}$.
Conversely, when $S \ufldp{p} T$ holds, 
nodes above $p$ in $T$ admit less sharing than nodes above $p$ in $S$.
Extending the graph rewrite relation $\REW[\GS(\RS),p]$ by $\flds{p}$
and $\ufldp{p}$ 
addresses both problems highlighted in Example \ref{ex:problem:1}
and Example \ref{ex:problem:2}.

The relations $\flds{p}$ and $\ufldp{p}$ are based 
on \emph{single step} approximations $\flda{u}{v}$ of $\gem$.
\begin{definition}\label{d:flda}
Let $\succ$ denote some total order on nodes that is irreflexive and transitive, let $\succcurlyeq$ 
denote the reflexive closure of $\succ$.
Let $S$ be term graphs, and let $u, v \in S$ be nodes satisfying $u \succcurlyeq v$.
We define $S \fldaeq{u}{v} T$ for term graph $T$ if 
$S \geqm T$ for the  morphism $m$ 
identifying $u$ and $v$, more precisely,
$m(u) = v$ and $m(w) = w$ for all $w \in S \setminus \set{u}$.
We define $S \flda{u}{v} T$ if $S \fldaeq{u}{v} T$ and $u \not = v$. 
\end{definition}  

We write $S \fldaeq{}{v} T$ ($S \flda{}{v} T$) if there exists $u \in S$ such that 
$S \fldaeq{u}{v} T$ ($S \flda{u}{v} T$) holds.
Similar $S \fldaeq{}{} T$ ($S \flda{}{} T$) if there exist nodes $u,v \in S$ such that 
$S \fldaeq{u}{v} T$ ($S \flda{u}{v} T$) holds.

\begin{example}\label{ex:flda}
  \newcommand{\fldtwothree}{\uflda{\text{\tiny{2}}}{\text{\tiny{3}}}}
  \newcommand{\fldfourfive}{\flda{\text{\tiny{4}}}{\text{\tiny{5}}}}
  \tikzstyle{nid}=[xshift=3mm,yshift=-1mm]
  \newcommand{\tnode}[4][]{%
    \node[#1] (#2) {#3};
    \node[nid] at (#2) {\tiny{\textbf{#4}}};
  }
  Consider the term $t = (\mZ \mPlus \mZ) \mTimes (\mZ \mPlus \mZ)$. Then
  $t$ is represented by the following three graphs that are related by 
  $\fldtwothree$ and $\fldfourfive$ respectively.
  \begin{center}
   \begin{tikzpicture}%
     [node distance=8mm %
     , bg/.style ={fill=black!3,draw=black,minimum width=2.cm}]

     \begin{scope}[xshift=-3.4cm]
       \tnode[]{A1}{$\cOne$}{$\mTimes$}
       \tnode[below of=A1]{A2}{$\cThree$}{$\mPlus$}
       \tnode[below of=A2, xshift=-5mm]{A3}{$\cFour$}{$\mZ$}
       \tnode[below of=A2, xshift=5mm]{A4}{$\cFive$}{$\mZ$}

       \path (A1) edge [bend left] (A2);
       \path (A1) edge [bend right] (A2);
       \path (A2) edge [] (A3);
       \path (A2) edge [] (A4);
       \node[] (L) at (0,-2.3)    {$T_{1}$};
       \begin{pgfonlayer}{background}
         \node [bg, fit=(A1) (A2) (A3) (A4) ] {};
       \end{pgfonlayer}
     \end{scope}

     \node at (-1.7,-0.8) {$\fldtwothree$};

     \begin{scope}
       \tnode[]{A1}{$\cOne$}{$\mTimes$}
       \tnode[below of=A1, xshift=-5mm]{A2}{$\cTwo$}{$\mPlus$}
       \tnode[below of=A1, xshift=5mm]{A3}{$\cThree$}{$\mPlus$}
       \tnode[below of=A2]{A4}{$\cFour$}{$\mZ$}
       \tnode[below of=A3]{A5}{$\cFive$}{$\mZ$}

       \path (A1) edge (A2);
       \path (A1) edge (A3);
       \path (A2) edge (A4);
       \path (A2) edge (A5);
       \path (A3) edge (A4);
       \path (A3) edge (A5);
       \node[] (L) at (0,-2.3)    {$T_{2}$};
       \begin{pgfonlayer}{background}
         \node [bg, fit=(A1) (A2) (A3) (A4) (A5)] {};
       \end{pgfonlayer}
     \end{scope}

     \node at (1.7,-0.8) {$\fldfourfive$};

     \begin{scope}[xshift=3.4cm]
       \tnode[]{A1}{$\cOne$}{$\mTimes$}
       \tnode[below of=A1, xshift=-5mm]{A2}{$\cTwo$}{$\mPlus$}
       \tnode[below of=A1, xshift=5mm]{A3}{$\cThree$}{$\mPlus$}
       \tnode[below of=A3, xshift=-5mm]{A4}{$\cFive$}{$\mZ$}

       \path (A1) edge (A2);
       \path (A1) edge (A3);
       \path (A2) edge [bend left] (A4);
       \path (A2) edge [bend right] (A4);
       \path (A3) edge [bend left] (A4);
       \path (A3) edge [bend right] (A4);
       \node[] (L) at (0,-2.3)    {$T_{3}$};
       \begin{pgfonlayer}{background}
         \node [bg, fit=(A1) (A2) (A3) (A4)] {};
       \end{pgfonlayer}
     \end{scope}
   \end{tikzpicture}
 \end{center} 
 Put otherwise, the term graph $T_2$ is obtained from $T_1$ by copying node $3$, 
 introducing the fresh node $2$. 
 The graph $T_3$ is obtained from $T_2$ by collapsing node $4$ onto node $5$.
\end{example}

Suppose $S \flda{u}{v} T$. Then the morphism
underlying $\flda{u}{v}$ defines the identity on $\nodes[S] \setminus \set{u}$. 
In particular, it defines the identity on successors of $u,v \in S$. 
Thus the following is immediate. 
\begin{lemma}\label{l:flda}
  Let $S$ be a term graph, and let $u,v \in S$ be two distinct nodes. 
  Then there exists a term graph $T$ such that $S \flda{u}{v} T$ 
  if and only if $\lab[S](u) = \lab[S](v)$ and $\suc[S](u) = \suc[S](v)$.
\end{lemma}
\begin{proof}
We prove the direction from left to right as the other is trivial.
Suppose $S \flda{u}{v} T$ and 
let $m$ be the morphism underlying $\flda{u}{v}$. 
Observe $m(u) = m(v) = v$. 
And thus
$\lab[S](u) = \lab[T](m(u)) = \lab[T](v) = \lab[T](m(v)) = \lab[S](v)$ follows.
To prove the second assertion, pick nodes $u_i$ and $v_i$ such that
$u \suci[S]{i} u_i$ and $v \suci[S]{i} v_i$.
Since $S \geqm T$ we obtain $m(u_i) = m(v_i)$. 
Thus $u_i \not = v_i$ if and only if either $u_i = u$ or $v_i = u$.
The former implies that $S$ is cyclic, the latter implies that $T$ is cyclic, 
contradicting that $S$ is a term graph or $m$ a morphism. We conclude 
$u_i = v_i$ for $i$ arbitrary, hence $\suc[S](u) = \suc[S](v)$.
\end{proof}

The restriction $u \succcurlyeq v$ was put onto $\fldaeq{u}{v}$ so that $\fldaeq{}{v}$ 
enjoys the following diamond property. Otherwise, 
the peak ${\uflda{u}{v} \cdot \flda{v}{u}} \subseteq {\isomorphic}$ 
cannot be joined.

\begin{lemma}\label{l:fldaeq:diamond}
  $
   {\ufldaeq{}{u} \cdot \fldaeq{}{v}}~\subseteq~{\fldaeq{}{w_1} \cdot \ufldaeq{}{w_2}}
  $
  where $w_1,w_2 \in \set{u,v}$.
\end{lemma}
\begin{proof}
Assume $T_1 \ufldaeq{u'}{u} S \fldaeq{v'}{v} T_2$ for some term graphs $S$, $T_1$ and $T_2$. 
The only non-trivial case is $T_1 \uflda{u'}{u} S \flda{v'}{v} T_2$ for $u' \not = v'$ and $u \not = v$.
We prove $T_1 \flda{}{w_1} \cdot \uflda{}{w_2} T_2$ for $w_1,w_2 \in \set{u,v}$ by case analysis.
\begin{itemize}
\item \icase{$T_1 \uflda{u'}{w} S \flda{v'}{w} T_2$ for $v' \not = u'$}
  We claim $T_1 \flda{v'}{w} \cdot \uflda{u'}{w} T_2$.  
  Let $m_1$ be the morphism underlying $\uflda{u'}{w}$ 
  and let $m_2$ be the morphism underlying $\flda{v'}{w}$ (c.f. Definition \ref{d:flda}).
  We first show
  $\lab[T_1](v') = \lab[T_1](w)$ and $\suc[T_1](v') = \suc[T_1](w)$.
  Using Lemma~\ref{l:flda},  $S \flda{v'}{w} T_2$ yields $\lab[S](v') = \lab[S](w)$.
  Employing $v' \not=u'$ and $w \not= u'$ we see 
  $$
  \lab[T_1](v') = \lab[T_1](m_1(v')) = \lab[S](v') = 
  \lab[S](w) = \lab[T_1](m_1(w)) = \lab[T_1](w)
  $$
  where we employ $m_1(v') = v'$ and $m_1(w) = w$.
  Again by Lemma~\ref{l:flda}, we see
  $\suc[S](u') = \suc[S](w)$ and $\suc[S](v') = \suc[S](w)$ by the assumption
  $T_1 \uflda{u'}{w} S \flda{v'}{w} T_2$.
  We conclude $\suc[S](v') = \suc[S](w)$ and thus
  $$
  \begin{array}{l}
    \suc[T_1](v') = \suc[T_1](m_1(v')) = m_1^*(\suc[S](v')) \\
    \hspace{45mm} = m_1^*(\suc[S](w)) = \suc[T_1](m_1(w)) = \suc[T_1](w) \tpkt
  \end{array}
  $$
  By Lemma~\ref{l:flda} we see $T_1 \flda{v'}{w} U_1$ for some term graph $U_1$.
  Symmetrically, we can prove $T_2 \flda{u'}{w} U_2$ for some term graph $U_2$.
  Hence $T_1 \flda{v'}{w} \cdot \uflda{u'}{w} T_2$ holds if $U_1 = U_2$. 
  To prove the latter, one shows $m_2 \cdot m_1 = m_1 \cdot m_2$
  by a straight forward case analysis. 
\item \icase{$T_1 \uflda{w}{u} S \flda{w}{v} T_2$ for $u \not = v$}
  Without loss of generality assume $u \succ v$. 
  We claim $T_1 \flda{u}{v} \cdot \uflda{u}{v} T_2$ and follow the pattern 
  of the proof for the previous case.
  Note that $\lab[T_1](u) = \lab[T_1](v)$ follows
  from $\lab[S](u) = \lab[S](w) = \lab[S](v)$ as before, 
  similar $\suc[T_1](u) = \suc[T_1](v)$ follows from 
  $\suc[S](u) = \suc[S](w) = \suc[S](v)$ with $w \not\in \suc[S](w)$.
  Hence $T_1 \flda{u}{v} U_1$ and symmetrically $T_2 \flda{u}{v} U_2$
  for some term graphs $U_1$ and $U_2$. 
  One verifies $m \cdot m_1 = m \cdot m_2$ for graph morphisms 
  $m_1$ underlying $\flda{w}{u}$, 
  $m_2$ underlying $\flda{w}{v}$, and
  $m$ underlying $\flda{u}{v}$.
  We conclude $T_1 \flda{u}{v} \cdot \uflda{u}{v} T_2$.
\item \icase{$T_1 \uflda{w}{u} S \flda{v'}{w} T_2$} 
  Note $v' \succ u$ since $v' \succ w \succ u$ by the assumption.
  We claim $T_1 \flda{v'}{u} \cdot \uflda{w}{u} T_2$.
  From the assumption we obtain 
  $\lab[S](u) = \lab[S](w) = \lab[S](v')$ and  $\suc[S](u) = \suc[S](w) = \suc[S](v')$, 
  from which we infer
  $\lab[T_1](v') = \lab[T_1](u)$ and $\suc[T_1](v') = \suc[T_1](u)$ (employing $w \not \in \suc[S](w)$).
  Further, $\lab[T_2](w) = \lab[T_2](u)$ and $\suc[T_2](w) = \suc[T_2](u)$   (employing $v' \not \in \suc[S](v')$).
  We conclude $T_1 \flda{v'}{u} U_1$ and similar $T_2 \flda{w}{u} U_2$. 
  Finally, one verifies $U_1 = U_2$ by case analysis.
\item \icase{$T_1 \uflda{u'}{u} S \flda{v'}{v} T_2$ for pairwise distinct $u',u,v'$ and $v$} 
  We show $T_1 \flda{v'}{v} \cdot \uflda{u'}{u} T_2$.
  Let $m$ be the morphism underlying $\flda{u'}{u}$.
  Observe $m(v) = v$ and $m(v') = v'$ by our assumption.
  Hence
  $\lab[T_1](v') = \lab[S](v') = \lab[S](v) = \lab[T_1](v)$ 
  and 
  $\lab[T_1](v') = m^*(\lab[S](v')) = m^*(\lab[S](v)) = \lab[T_1](v)$.
  We obtain $T_1 \flda{v'}{v} U_1$ and symmetrically $T_2 \flda{u'}{u} U_2$
  for some term graphs $U_1$ and $U_2$. 
  Finally, one verifies $U_1 = U_2$ by case analysis as above. 
\end{itemize}
\end{proof}
The above lemma implies confluence of $\fldaeq{}{}$.
Since ${\flda{*}{}} = {\fldaeq{*}{}}$, also $\flda{}{}$ is confluent.

\begin{definition}
  Let $S$ be a term graph and let $p$ be a position in $S$.
  We say that \emph{$S$ folds strictly below $p$ to the term graph $T$}, 
  in notation $S \flds{p} T$, 
  if $S \flda{u}{v} T$ for nodes $u,v \in S$ strictly below $p$ in $S$.
  The graph $S$ \emph{unfolds above $p$ to the term graph $T$}, 
  in notation $S \ufldp{p} T$, if 
  $S \uflda{u}{v} T$ for some unshared node $u \in T$ above $p$, i.e., 
  $\Pos[T](u) = \set{q}$ for $q \leqslant p$.
\end{definition}

\begin{example}\label{ex:fld:pos}
  Reconsider the term graphs $T_1$, $T_2$ and $T_3$ with 
  $T_1 \uflda{\text{\tiny{2}}}{\text{\tiny{3}}} T_2 \flda{\text{\tiny{4}}}{\text{\tiny{5}}} T_3$
  from Example \ref{ex:flda}.
  Then $T_1 \ufldp{2} T_2$ since node $3$ is an unshared node above position $2$ in $T_2$.
  Further $T_2 \flds{2} T_3$ since both nodes $4$ and $5$ are strictly below 
  position $2$ in $T_2$.
\end{example}

Note that for $S \flda{u}{v} T$ the sets of positions  $\Pos[S]$ and $\Pos[T]$ coincide, 
thus the $n$-fold composition $\ufldp{p}^n$ of $\ufldp{p}$ 
 (and the $n$-fold composition $\flds{p}^n$ of $\flds{p}$) is well-defined for $p \in \Pos[S]$.
In the next two lemmas we prove that relations $\ufldp{p}$ and $\flds{p}$ fulfill their intended purpose.

\begin{lemma}\label{l:ufldp:unshared}
  Let $S$ be a term graph and $p$ a position in $S$.
  If $S$ is $\ufldp{p}$-minimal 
  then the node corresponding to $p$ is unshared.
\end{lemma}
\begin{proof}
By way of contradiction, suppose $S$ is $\ufldp{p}$-minimal 
but the node $w$ corresponding to $p$ is shared.
We construct $T$ such that $S \ufldp{p} T$.
We pick an unshared node $v \in S$,
and shared node $v_i \in S$, above $p$ such that $v \reach[S] v_i$.
By a straight forward induction on $p$ we see that
$v$ and $v_i$ exist as $w$ is shared. For this, note that 
at least the root of $S$ is unshared.

Define $T \defsym (\nodes[T], \lab[T], \suc[T])$ as follows:
let $u$ be a fresh node such that $u \suc v_i$.
set $\nodes[T] \defsym \nodes[S] \cup \set{u}$;
set $\lab[T](u) \defsym \lab[S](v_i)$ and $\suc[T](u) \defsym \suc[S](v_i)$;
further replace the edge $v \suci{i} v_i$ by $v \suci{i} u$, that is, 
set $\lab[T](v) \defsym [v_1, \dots, u, \dots, v_l]$ 
for $\lab[S](v) = [v_1, \dots, v_i, \dots, v_l]$.
For the remaining cases, define $\lab[T](w) \defsym \lab[S](u)$ and
$\suc[T](w) \defsym \suc[S](w)$.
One easily verifies $T \flda{u}{v_i} S$. 
Since by way of construction $u$ is an unshared node above $p$, $S \ufldp{p} T$ holds.
\end{proof}

\begin{lemma}\label{l:flds:maxshared}
  Let $S$ be a term graph, let $p$ be a position in $S$.
  If $S$ is $\flds{p}$-minimal then the subgraph
  $\subgraphAt{S}{p}$ is maximally shared.
\end{lemma}
\begin{proof}
Suppose $\subgraphAt{S}{p}$ is not maximally shared. 
We show that $S$ is not $\flds{p}$-minimal.
Pick some node $u \in \subgraphAt{S}{p}$ such that
there exists a distinct node $v \in \subgraphAt{S}{p}$ with 
$\trepr(\subgraphAt{S}{u}) = \trepr(\subgraphAt{S}{v})$.
For that we assume that $u$ is \emph{$\reach[S]$-minimal} in 
the sense
that there is no node $u'$ with $u \reachtir u'$
such that $u'$ would fulfill the above property.
Clearly $\lab[S](u) = \lab[S](v)$ follows from
$\trepr(\subgraphAt{S}{u}) = \trepr(\subgraphAt{S}{v})$. 
Next, suppose $u \suci{i} u_i$ and $v \suci{i} v_i$
for some nodes $u_i \not = v_i$.
But then $u_i$ contradicts minimality of $u$, and so we conclude
$u_i = v_i$. Consequently $\suc[S](u) = \suc[S](v)$ follows as desired.
Without loss of generality, suppose $u \succ v$.
By Lemma~\ref{l:flda}, $S \flda{u}{v} T$ for some term graph $T$,
since $u,v \in \subgraphAt{S}{p}$ also $S \flds{p} T$ holds.
\end{proof}

\begin{theorem}[Adequacy]\label{t:adequacy}
Let $s$ be a term and let $S$ be a term graph such that $\trepr(S) = s$. Then
$${s} \rew[\RS,p] {t} \text{ if and only if } S \ufldp{p}^! \cdot \flds{p}^! \cdot  \REW[\GS(\RS),p] T \tkom$$ 
for some term graph $T$ with $\trepr(T) = t$.
\end{theorem}
\begin{proof}
First, we consider the direction from right to left.
Assume $S \ufldp{p}^! U \flds{p}^! V  \REW[\GS(\RS),p] T$.
Note that $\flds{p}$ preserves $\ufldp{p}$-minimality.
We conclude $V$ is $\ufldp{p}$-minimality as $U$ is. 
Let $v \in V$ be the node corresponding to $p$. 
By Lemma~\ref{l:ufldp:unshared} we see $\Pos[U](v) = \set{p}$.
Now consider the step $V  \REW[\GS(\RS),p] T$. 
There exists a renaming ${L' \to R'}$ of 
${(L \to R)} \in \GS(\RS)$ such that $m \colon L' \to \subgraphAt{V}{v}$
is a morphism and $T = \replaceAt{V}{v}{m(R')}$.
Set $l \defsym \trepr(L')$ and $r \defsym \trepr(R')$, by definition ${(l \to r)} \in \RS$.
By Lemma~\ref{l:subst:l} we obtain $l\sigma_m = \trepr(\subgraphAt{V}{v})$
for the substitution $\sigma_m$ induced by the morphism $m$. 
Define the context $C \defsym \trepr(\replaceAt{V}{v}{\hole})$. 
As $v$ is unshared, $C$ admits exactly one occurrence of $\hole$, 
moreover the position of $\hole$ in $C$ is $p$.
By Lemma~\ref{l:context},
$$
\trepr(V) 
= \trepr(\replaceAt{V}{v}{\subgraphAt{V}{v}})
= C[\trepr(\subgraphAt{V}{v})]
= C[l\sigma_m] \tpkt
$$
Set $T_v \defsym \subgraphAt{(m(R') \cup V)}{\grt(m(R'))}$, 
and observe $T = \replaceAt{V}{v}{m(R')} = \replaceAt{V}{v}{T_v}$.
Using Lemma~\ref{l:context} and Lemma~\ref{l:subst:r} we see
$$
\trepr(T)
= \trepr(\replaceAt{V}{v}{T_v})
= C[\trepr(T_v)]
= C[r\sigma_m] \tpkt
$$
As $\trepr(S) = \trepr(V)$ by Lemma~\ref{l:morph:termeq}, 
$\trepr(S) = C[l\sigma_m] \rew[\RS,p] C[r\sigma_m] = \trepr(T)$ follows.

Finally, consider the direction from left to right.
For this suppose 
$
s = C[l\sigma] \rew[\RS,p] C[r\sigma] = t
$ 
where the position of the hole in $C$ is $p$. 
Suppose $S \ufldp{p}^! U \flds{p}^! V$ for $\trepr(S) = s$. We prove 
that there exists $T$ such that 
$V \REW[\GS(\RS),p] T$ and $\trepr(T) = t$. 
Note that $V$ is $\flds{p}$-minimal and, as observed above, 
it is also $\ufldp{p}$-minimal. 
Let $v \in V$ be the node corresponding to $p$, 
by Lemma~\ref{l:ufldp:unshared} the node $v$ is unshared.
Next, observe $l\sigma = \subtermAt{s}{p} = \trepr(\subgraphAt{S}{p}) = \trepr(\subgraphAt{V}{v})$ 
since $\trepr(S) = \trepr(V)$ (c.f. Lemma~\ref{l:morph:termeq}).
Additionally, Lemma~\ref{l:flds:maxshared} reveals $\subgraphAt{V}{v} \in \Shared(l\sigma)$.
Further, by Lemma~\ref{l:context} we see
$$
s = C[l\sigma] 
= \trepr(V) 
= \trepr(\replaceAt{V}{v}{\subgraphAt{V}{v}})
= \trepr(\replaceAt{V}{v}{\hole})[l\sigma]  \tpkt
$$
Since the position of the hole in $C$ and $\trepr(\replaceAt{V}{v}{\hole})$
coincides, we conclude that $C = \trepr(\replaceAt{V}{v}{\hole})$.

Let ${L \to R} \in \SGRS(\RS)$ be the rule corresponding to $(l \to r) \in \RS$, 
let $(L' \to R') \isomorphic (L \to R)$ be a renaming for $V$.
As $L' \in \Tree(l)$ and $\subgraphAt{V}{v} \in \Shared(l\sigma)$, 
by Lemma~\ref{l:match:lhs} there exists a morphism $m \colon L' \to \subgraphAt{V}{v}$
and hence $V \REW[\GS(\RS),p] T$ for $T = \replaceAt{V}{v}{m(R')}$.
Note that for the induced substitution $\sigma_m$ and $x \in \Var(l)$, 
$\sigma_m(x) = \sigma(x)$. 
Set $T_v \defsym \subgraphAt{(m(R') \cup V)}{\grt(m(R'))}$, 
hence $T = \replaceAt{V}{v}{T_v}$ and moreover $r\sigma = r\sigma_m = \trepr(T_v)$
follows as in the first half of the proof.
Employing Lemma~\ref{l:context} we obtain
$$
t
= C[r\sigma]
= \trepr(\replaceAt{V}{v}{\hole})[r\sigma]
= \trepr(\replaceAt{V}{v}{T_v})
= \trepr(T)\tpkt 
$$
\end{proof}

We define $S \REWS[\GS(\RS),p] T$ if and only if 
${S} \ufldp{p}^! \cdot \flds{p}^! {U} \REW[\GS(\RS),p] {T}$.
Employing this notion we can rephrase the conclusion of the Adequacy Theorem 
as: ${s} \rew[\RS,p] {t}$ if and only if $S \REWS[\GS(\RS),p] T$
for $\trepr(S) = s$ and $\trepr(T) = t$.


\section{Implementing Term Rewriting Efficiently} \label{Complexity}
Opposed to term rewriting, graph rewriting 
induces linear size growth in the length of derivations. 
The latter holds as a single step $\REW[\GS]$ admits constant size growth:
\begin{lemma}\label{l:sizebounds:simple}
  If $S \REW[\GS] T$ then $\size{T} \leqslant \size{S} + \Delta$ for 
  some $\Delta \in \NAT$ depending only on $\GS$.
\end{lemma} 
\begin{proof}
  Set $\Delta \defsym \max\set{\size{R} \mid {(L \to R)} \in \GS}$
  and the lemma follows by definition.
\end{proof}
It is easy to see that a graph rewrite step $S \REW[\GS] T$ can be performed in time polynomial
in the size of the term graph $S$. By the above lemma we obtain that $S$ can be normalised 
in time polynomial in $\size{S}$ and the length of derivations. 
In the following, we prove a result similar to Lemma~\ref{l:sizebounds:simple} for the relation 
$\REWS[\GS]$, where (restricted) folding and unfolding is incorporated.
The main obstacle is that due to unfolding, size growth of $\REWS[\GS]$ is not bound by a constant in general. 
We now investigate into the relation $\ufldp{p}$ and $\flds{p}$. 
\begin{lemma}\label{l:fold:bound}
  Let $S$ be a term graph and let $p$ a position in $S$. 
  \begin{enumerate}
  \item\label{l:fold:bound:1} If $S \ufldp{p}^\ell T$ then $\ell \leqslant \size{p}$ and 
    $\size{T} \leqslant \size{S} + \size{p}$.
  \item\label{l:fold:bound:2} If $S \flds{p}^\ell T$ then $\ell \leqslant \size{\subgraphAt{S}{p}}$
    and $\size{T} \leqslant \size{S}$.
  \end{enumerate}
\end{lemma}
\begin{proof}
We consider the first assertion.
For term graphs $U$, let
$P_U = \{w \mid \Pos[U](w) = \{q\} \text{ and } q \leqslant p\}$
be the set of unshared nodes above $p$.
Consider $U \ufldp{p} V$.
Observe that $P_U \subset P_V$ holds:
By definition $U \uflda{u}{v} V$ where $\Pos[V](u) = \set{q}$ with $q \leqslant p$.
Clearly, $P_U \subseteq P_V$, but moreover $u \in P_V$ whereas $u \not \in P_U$. 
Hence for $(S \ufldp{p}^\ell T) = S = S_0 \ufldp{p} \dots \ufldp{p} S_\ell = T$, 
we observe $P_S = P_{S_0} \subset \dots P_{S_\ell} = P_T$.
Note that $\size{P_S} \geqslant 1$ since $\grt(S) \in P_s$. 
Moreover, $\size{P_T} = \size{p} + 1$ since the node corresponding to $p$ in $T$ is unshared (c.f. Lemma~\ref{l:ufldp:unshared}).
Thus from $P_{S_i} \subset P_{S_{i+1}}$ ($0 \leqslant i < \ell$) we conclude $\ell \leqslant \size{p}$.
Next, we see $\size{T} \leqslant \size{S} + \size{p}$ as 
$\size{T} = \size{S} + \ell$ by definition of $\ufldp{p}$.

We now prove the second assertion.
Consider term graphs $U$ and $V$ such that 
$U \flds{p} V$. By definition 
$U \flda{u}{v} V$ where nodes $u$ and $v$ are strictly below position $p$ in $U$.
Hence $U \geqm V$ for the morphism $m$ underlying $\flda{u}{v}$. 
As a simple consequence of Lemma~\ref{l:morph:subgraph}, 
we obtain $\subgraphAt{U}{p} \gem \subgraphAt{V}{p}$
and thus $\size{\subgraphAt{U}{p}} > \size{\subgraphAt{V}{p}}$.
From this we conclude the lemma as above, where for 
$\size{T} \leqslant \size{S}$ we employ that 
if $U \flda{u}{v} V$ then $\size{\subgraphAt{V}{p}} = \size{\subgraphAt{U}{p}} - 1$.
\end{proof}
By combining the above two lemmas we derive the following:
\begin{lemma} \label{l:sizebounds}
  If $S \REWS T$
  then $\size{T} \leqslant \size{S} + \depth(S) + \Delta$ and 
  $\depth(T) \leqslant \depth(S) + \Delta$ for some $\Delta \in \NAT$ depending only on $\GS$.
\end{lemma} 
\begin{proof}
Consider $S \REWS T$, i.e., ${S} \ufldp{p}^! U \flds{p}^! V \REW[\GS] {T}$ for 
some position $p$ and term graphs $U$ and $V$. 
Lemma~\ref{l:fold:bound} reveals $\size{U} \leqslant \size{S} + \size{p}$ and
further $\size{V} \leqslant \size{U}$
for $\Delta \defsym \max\set{\size{R} \mid {(L \to R)} \in \GS}$.
As $\size{p} \leqslant \depth(S)$ we see $\size{V} \leqslant \size{S} + \depth(S)$.
Since $V \REW T$ implies $\size{T} \leqslant \size{V} + \Delta$ (c.f. Lemma~\ref{l:sizebounds:simple})
we establish $\size{T} \leqslant \size{S} + \depth(S) + \Delta$.
Finally, $\depth(T) \leqslant \depth(S) + \Delta$ follows from the easy observation that
both $U \ufldp{p} V$ and $U \flds{p} V$ imply $\depth(U) = \depth(V)$,
likewise $V \REW[\GS] {T}$ implies $\depth(T) \leqslant V + \Delta$.
\end{proof}

\begin{lemma}\label{l:space}
  If $S \REWSL{\ell} T$ then $\size{T} \leqslant (\ell + 1) \size{S} + \ell^2 \Delta$
  for $\Delta \in \NAT$ depending only on $\GS$.
\end{lemma} 
\begin{proof}
We prove the lemma by induction on $\ell$. The base case follows trivially, 
so suppose the lemma holds for $\ell$, we establish the lemma for $\ell + 1$. 
Consider a derivation $S \REWSL{\ell} T \REWS U$. 
By induction hypothesis, 
$\size{T} \leqslant (\ell + 1) \size{S} + \ell^2 \Delta$. 
Iterative application of 
Lemma~\ref{l:sizebounds} reveals
$\depth(T) \leqslant \depth(S) + \ell\Delta$.
Thus
\begin{align*}
  \size{U} & \leqslant \size{T} + \depth(T) + \Delta \\ 
  & \leqslant \bigl((\ell + 1) \size{S} + \ell^2 \Delta\bigr) + \bigl(\depth(S) + \ell\Delta\bigr) + \Delta \\ 
  & \leqslant (\ell + 2) \size{S} + \ell^2\Delta + \ell\Delta + \Delta \\
  & \leqslant (\ell + 2) \size{S} + (\ell + 1)^2 \Delta  \tpkt
\end{align*}
\end{proof}

In the sequel, we prove that an arbitrary graph rewrite step $S \REWS[] T$
can be performed in time cubic in the size of $S$.  
Lemma~\ref{l:space} then allows us to lift the bound on steps
to a polynomial bound on derivations in the size of $S$ and the length of derivations.
We closely follow the notions of \cite{Kozen}. 
As model of computation we use $k$-tape \emph{Turing Machines} (TM for short) with
dedicated input- and output-tape. If not explicitly mentioned otherwise, 
we will use deterministic TMs.
We say that a (possibly nondeterministic) TM computes a relation
$R \subseteq \Sigma^* \times \Sigma^*$ if for all $(x,y) \in R$, on input $x$ there 
exists an accepting run such that $y$ is written on the output tape.

We fix a \emph{standard encoding} for term graphs $S$. 
We assume that for each function symbol $f \in \FS$ a corresponding tape-symbols is present. 
Nodes and variables are represented by natural numbers, encoded in binary notation and possibly padded by zeros.
We fix the invariant that natural numbers $\set{1, \dots,\size{S}}$ are used 
for nodes and variables in the encoding of $S$.
Thus variables (nodes) of $S$ are representable in space $\bigO(\log(\size{S}))$.
Finally, term graphs $S$ are encoded as a \emph{ordered} list of \emph{node specifications}, i.e., 
triples of the form  $\tuple{v,\Lab(v),\suc(v)}$ for all 
$v \in S$ (compare \cite[Section 13.3]{TeReSe}).
We call the entries of a node specification
\emph{node-field}, \emph{label-field} and \emph{successor-field} respectively.
We additionally assume that each node specification has a status flag (constant in size) attached.
For instance, we use this field below to mark subgraphs.

For a suitable encoding of tuples and lists, 
a term graph $S$ is representable in size $\bigO(\log(\size{S})  * \size{S})$. 
For this, observe that the length of $\suc(v)$ is bound by the maximal arity of the fixed signature $\FS$.
In this spirit, we define the \emph{representation size} of a term graph $S$ as $\rsize{S} \defsym \bigO(\log(\size{S}) * \size{S})$. 

Before we investigate into the computational complexity of $\REWS[]$, 
we prove some auxiliary lemmas.
\begin{lemma}\label{l:comp:marking}
  The subgraph $\subgraphAt{S}{u}$ of $S$ can be marked in quadratic time in $\rsize{S}$.
\end{lemma}
\begin{proof}
We use a TM that operates as follows:
First, the graph $S$ is copied from the input tape to a working tape. Then 
the node specifications of $v \in \subgraphAt{S}{u}$ are marked in a breath-first manner.
Finally the resulting graph is written on the output tape.
In order to mark nodes two flags are employed, the \emph{permanent} and the \emph{temporary} flag. 
A node $v$ is marked permanent by marking its node specification permanent, 
and all node specifications of $\suc(v)$ temporary.
Initially, $u$ is marked.
Afterward, the machine iteratively marks temporary marked node permanent, 
until all nodes are marked permanent.

Notice that a node $v$ can be marked in time linear in $\rsize{S}$. 
For that, the flag of the node specification of $v$ is set appropriately, 
and $\suc(v)$ is copied on a second working tape.
Then $S$ is traversed, and the node-field of each encountered node specification
is compared with the current node written on the second working. If the nodes coincide, 
then the flag of the node specification is adapted, the pointer on the second working tape
advances to the next node, and the process is repeated. Since the length of $\suc(v)$ is
bounded by a constant (as $\FS$ is fixed), marking a node requires a constant 
number of iterations.  
Since the machine marks at most $\size{S} \leqslant \rsize{S}$ nodes 
we obtain an overall quadratic time bound in $\rsize{S}$.
\end{proof}

\begin{lemma}\label{l:comp:ufldp}
  Let $S$ be a term graph and let $p$ a position in $S$. 
  A term graph $T$ such that $S \ufldp{p}^{!} T$ is 
  computable in time $\bigO(\rsize{S}^2)$.
\end{lemma}
\begin{proof}
Given term graph $S$, we construct a TM that produces 
a term graph $T$ such that $S \ufldp{p}^{!} T$ in time quadratic in $\rsize{S}$.
For that, the machine traverse $S$
along the path induced by $p$ and introduces a fresh copy for each
shared node encountered along that path. 
The machine has four working tapes at hand. 
On the first tape, the graph $S$ is copied in such a way 
that nodes are padded sufficiently by leading zero's so that
successors can be replaced by fresh nodes $u \leqslant 2 * \size{S}$ inplace.
The graph represented on the first tape is called the \emph{current graph}, its size
will be bound by $\bigO{\rsize{S}}$ at any time.
On the second tape the position $p$, encoded as list of argument positions, is copied. 
The argument position referred by the tape-pointer is called \emph{current argument position}
and initially set to the first position. 
The third tape holds the \emph{current node}, initially the root $\grt(S)$ of $S$.
Finally, the remaining tape holds the size of the current graph in binary, the 
\emph{current size}. 
One easily verifies that these preparatory steps can be done in time linear in $S$. 

The TM now iterates the following procedure, 
until every argument position in $p$ was considered.
Let $v$ be the current node, let $S_i$ the current graph 
and let $i$ be the current argument position. 
We machine keeps the invariant that $v$ is unshared in $S_i$.
First, the node $v_i$ with $v \suci{i} v_i$ in $S_i$ is determined
in time linear in $\rsize{S}$, the current node is replaced by $v_i$. 
Further, the pointer on the tape holding $p$ is advanced to the next argument position.
Since $v$ is unshared, $v_i$ is shared if and only if $v_i \in \suc(u)$ for 
$u \not = v$. 
The machine checks whether $v_i$ is shared in the current graph, by the above observation
in time linear in $\rsize{S}$.
If $v_i$ is unshared, the machine enters the next iteration. 
Otherwise, the node $v_i$ is cloned in the following sense. 
First, the the $i$-th successor $v_i$ of $v$ is replaced by a fresh node $u$.
The fresh node is obtained by increasing the current node by one, 
this binary number is used as fresh node $u$. 
Further, the node specification $\langle u, \Lab(v_i), \suc(v_i) \rangle$
is appended to the current graph $S_i$.
Call the resulting graph $S_{i + 1}$. Then $S_i \uflda{u}{v_i} S_{i+1}$ with
$\Pos[S_{i+1}](u) = \set{q}$ and $q \leqslant p$, i.e., $S_i \ufldp{p} S_{i+1}$.

When the procedure stops, the machine has computed $S = S_0 \ufldp{p} S_1 \ufldp{p} \dots \ufldp{p} S_n = T$.
One easily verifies that $S_n$ is $\ufldp{p}$-minimal as every considered node along the path $p$
is unshared. Each iteration takes time linear in $\rsize{S}$. As 
as at most $\size{p} \leqslant \size{S}$ iterations have to be performed,
we obtain the desired bound.
\end{proof}

\begin{lemma}\label{l:comp:flds}
  Let $S$ be a term graph and $p$ a position in $S$. 
  The term graph $T$ such that $S \flds{p}^{!} T$ is 
  computable in time $\bigO(\rsize{S}^2)$.
\end{lemma}
\begin{proof}
\newcommand{\fldh}[2][]{\flda{#1}{\text{\tiny{\ensuremath{(#2)}}}}}
Define the \emph{height} $\height[U](u)$ of a node $u$ in a term graph $U$ inductively as usual:
$\height[U](u) \defsym 0$ if $\suc(u) = []$ and $\height[U](v) \defsym 1 + \max_{v \in \suc(u)} \height[U](v)$ otherwise.
We drop the reference to the graph $U$ when referring to the height of nodes
in the analysis of the normalising sequence $S \flds{p}^{!} T$ below. 
This is justified as the height of nodes remain stable under 
$\flda{}{}$-reductions. 

Recall the definition of $\flds{p}$: $U \flds{p} V$ if 
there exist nodes $u,v$ strictly below $p$ with $U \flda{u}{v} V$. 
Clearly, for $u,v$ given, the graph $V$ is constructable from $U$ in time linear 
in $\size{U}$. However, finding arbitrary nodes $u$ and $v$ such that
$U \flda{u}{v} V$ takes time quadratic in $\size{U}$ worst case.
Since up to linear many $\flda{}{}$-steps in $\size{S}$ need to be performed, a straight forward
implementation admits cubic runtime complexity. 
To achieve a quadratic bound in the size of the starting graph $S$, 
we construct a TM that implements a bottom up reduction-strategy. 
More precise, the machine implements the 
maximal sequence
\begin{equation}
  \label{l:comp:flds:I}
  \tag{a}
  S = S_1 \flda{!}{u_1} S_2 \flda{!}{u_2} \cdots \flda{!}{u_{\ell - 1}} S_\ell
\end{equation}
satisfying, for all $1 \leqslant i < \ell - 1$, 
(i) either $\height(u_i) = \height(u_{i+1})$ and $u \prec v$ or $\height(u_i) < \height(u_{i+1})$, and 
(ii) for $S_{i} \flda{v_{i,1}}{u_i} \dots \flda{v_{i,k}}{u_i} S_{i+1}$,
$u_i$ and $v_{i,j}$ ($1 \leqslant j \leqslant k$) are strictly below $p$.

By definition $S \flds{p}^{*} S_\ell$, it remains to show that the sequence 
\eqref{l:comp:flds:I} is normalising, i.e., $S_\ell$ is $\flds{p}$-minimal.
Set $d \defsym \depth(\subgraphAt{S}{p})$ and define, for 
$0 \leqslant h \leqslant d$, 
$$
{\fldh{h}} \defsym {\bigcup_{u,v \in \subgraphAt{S}{p} \wedge \height(v)=h} {\flda{u}{v}}} \tpkt
$$
Observe that each $\flda{}{u_i}$-step in the sequence \eqref{l:comp:flds:I} corresponds 
to a step $\fldh{h}$ for some $0 \leqslant h \leqslant d$.
Moreover, it is not difficult to see that
\begin{equation}
  \label{l:comp:flds:II}
  \tag{b}
  S = S_{i_0} \fldh[!]{0} S_{i_1} \fldh[!]{1} \cdots \fldh[!]{d} S_{i_{d+1}}= S_\ell
\end{equation}
for $\{S_{i_0}, \dots, S_{i_d+1}\} \subseteq \{\seq[\ell - 1]{S}\}$.

Next observe
$S_i \fldh{h_1} S_{i+1} \fldh{h_2} S_{i+2}$ and $h_1 > h_2$ 
implies $S_i \fldh{h_2} \cdot \fldh{h_1} S_{i+2}$:
suppose $S_i \flda{u'}{u} S_{i+1} \flda{v'}{v} S_{i+2}$ where $\height(u) > \height(v)$ and $u',u,v,v' \in \subgraphAt{S}{p}$, 
we show $S_i \flda{v'}{v} \cdot \flda{u'}{u} S_{i+2}$.
Inspecting the proof of Lemma~\ref{l:fldaeq:diamond} we see 
${\uflda{u'}{u} \cdot \flda{v'}{v}}~\subseteq~{\flda{v'}{v} \cdot \uflda{u'}{u}}$ 
for the particular case that $u',u,v$ and $v'$ pairwise distinct. 
The latter holds as $\height(u') = \height(u) \not=\height(v) = \height(v')$.
Hence it remains to show
$S_i \flda{v'}{v} S_{i+1}'$ for some term graph $S_{i+1}'$, or equivalently 
$\lab[S_i](v) = \lab[S_i](v')$ and $\suc[S_i](v) = \suc[S_i](v')$ by Lemma~{\ref{l:flda}}. 
The former equality is trivial, 
for the latter observe
$\height(u') = \height(u) > \height(v) = \height(v')$ and thus 
neither $u' \not \in \suc[S_i](v')$ nor $u' \not \in \suc[S_i](v)$. 
We see $\suc[S_i](v) = \suc[S_{i+1}](v) = \suc[S_{i+1}](v') = \suc[S_i](v')$.

Now suppose that $S_\ell$ is not $\flds{p}$-minimal, i.e, 
$S_\ell \fldh{h} U$ for some $0 \leqslant h \leqslant d$ and term graph $U$. 
But then we can permute steps in the reduction \eqref{l:comp:flds:II} such that
$S_{i_{h+1}} \fldh{h} V$ for some term graph $V$. This 
contradicts $\fldh[!]{h}$-minimality of $S_{i_{h+1}}$.
We conclude that $S_{\ell}$ is $\flds{p}$-minimal.
Thus sequence \eqref{l:comp:flds:I} is $\flds{p}$-normalising. 

We now construct a TM operating in time $\bigO(\rsize{S}^2)$ that,
on input $S$ and $p$, computes the sequence \eqref{l:comp:flds:I}.
We use a dedicated working tape to store the \emph{current graph} $S_i$.
Initially, the term graph $S$ is copied on this working tape.
Further, the node $w$ corresponding to $p$ in $S$ is computed by recursion on $p$ in time $\rsize{S}^2$.
Afterward, the quadratic marking algorithm of Lemma~\ref{l:comp:marking} is used to 
mark the subgraph $\subgraphAt{S}{w}$ in $S$. 

The TM operates in stages, 
where in each stage the current graph $S_{i_{h}}$ is replaced 
by $S_{i_{h+1}}$ for  $S_{i_h} \fldh[!]{h} S_{i_{h+1}}$.
Consider the subsequence 
\begin{equation}
  \label{l:comp:flds:III}  
  \tag{c}
  S_{i_h} = S_{j_1} \flda{!}{u_{j_1}}  \cdots \flda{!}{u_{j_l}} S_{j_{l+1}} = S_{i_{h+1}}
\end{equation}
of sequence \eqref{l:comp:flds:I} for $j_1 \defsym i_h$ and $j_{l+1} \defsym i_{h+1}$.
Then $h = \height(u_{j_1}) = \dots = \height(u_{j_l})$. Call $h$ the \emph{current height}.
To compute the above sequence efficiently, 
the TM uses 
the flags \emph{deleted}, \emph{temporary} and \emph{permanent}
besides the subterm marking.
Let $S_j$ ($j_1 \leqslant j \leqslant j_{l+1}$) be the current graph.
If a node $u$ is marked deleted, it is treated as if $u \not \in S_j$, that is, 
when traversing $S_j$ the corresponding node specification is ignored.
Further, the machine keeps the invariant that when $u \in \subgraphAt{S_j}{p}$ 
then (the node specification of) 
$u$ is marked permanent if and only if $\height(u) < h$.
Thus deciding whether $\height(u) = h$ for some node $u \in S_i$ amounts to 
checking whether $u$ is not marked permanent, but all successors $\suc(u)$ are marked permanent. 
To decide $\height(u) = h$ solely based on the node specification of $u$, 
the machine additionally record 
whether $u_i \in \suc(u)$ is marked permanent in the node specification of $u$. 
Since the length of $\suc(u)$ is bounded by a constant, this is can be done in constant space. 
At the beginning of each stage, the machine is in one of two states, 
say $\mathsf{p}$ and $\mathsf{q}$ 
(for current height $h = 0$, the initial state is $\mathsf{p}$).
\begin{itemize}
\item \textsc{State} $\mathsf{p}$.
  In this state the machine is searching the next node $u_j$ to collapse. 
  It keeps the invariant that previously considered nodes $u_{j_i}$ for $j_1 \leqslant j_i \leqslant i$ are marked temporary.
  Reconsider the definition of the sequence \eqref{l:comp:flds:I}.
  The node $u_j$ is the least node (with respect to $>$ underlying $\flda{}{}$)
  satisfying
  (i) $u_j$ is marked by the subterm marking, and
  (ii) $u_j$ is not marked permanent but all nodes in $\suc(u_j)$ are marked permanently, and
  (iii) $u_j$ is not marked temporary.
  Recall that node specifications are ordered in increasing order.
  In order to find $u_j$, the graph $S_j$ is scanned from top to bottom, 
  solely based on the node specification properties (i) --- (iii) are checked, 
  and the first node satisfying (i) -- (iii) is returned.

  Suppose the node $u_j$ is found. 
  The machine marks the node $u_j$ temporary and writes $u_j$, $\Lab(u_j)$ and 
  $\suc(u_j)$ on dedicated working tapes. Call $u_j$ the \emph{current node}.
  The machine goes into state $\mathsf{q}$ as described below.
  On the other hand, it $u_j$ is not found, the stage is completed as all nodes of 
  height $h$ are temporary marked.
  The temporary marks, i.e., the marks of node $u_{j_1}, \dots, u_{j_l}$ are transformed into permanent ones
  and the stage is completed. Notice that all nodes of height less or equal to $h$ are marked permanent this way.
  The invariant on permanent marks is recovered, the
  machine enters the next stage.

  One verifies that one transition from state $\mathsf{p}$
  requires at most linearly many steps in $\rsize{S}$. 
\item \textsc{State} $\mathsf{q}$.
  The machine iteratively computes the sequence 
  $$
  S_{j} = S_{1,j} \flda{v_{1,j}}{u_j} \dots \flda{v_{k-1,j}}{u_j} S_{k,j} = S_{j+1}
  $$
  for current node $u_j$ as determined in state $\mathsf{p}$.
  Suppose $S_{i,j}$, $1 \leqslant i \leqslant k$ is the current graph.
  The machine searches for the node $v_{i,j} \in \subgraphAt{S_{i,j}}{p}$, $v_{i,j} \geqslant v_{i-1,j}$ (for $i \geqslant 1$)
  such that 
  $S_{i,j} \flda{v_{i,j}}{u} S_{i,j + 1}$ for current node $u_j \in \subgraphAt{S_{i,j}}{p}$. 
  For that, the machine scans the current graph from top to bottom, comparing 
  label- and successor-field with the ones 
  written on the dedicated working tapes in state $\mathsf{p}$. 
  Then $v_{i,j} \in \subgraphAt{S_{i,j}}{p}$ is checked
  using the subterm marking.
  If $v_{i,j}$ is not found, the current graph $S_{i,j}$ is $\flda{}{u_j}$-minimal 
  according to Lemma~\ref{l:flda}. 
  The above sequence has been computed, 
  the machine enters state $\mathsf{p}$. 
  Otherwise, the machine writes $v_{i,j}$ on an additional working 
  tape and applies the morphism underlying $\flda{v_{i,j}}{u_j}$ on the current graph. 
  For that the specification of $v_{i,j}$ is marked as deleted and simultaneously 
  every occurrence of $v_{i,j}$ in successor-field of node specifications is replaced by $u_j$.
  The machine enters state $\mathsf{q}$ again.
  One verifies that one transition from state $\mathsf{q}$ to either $\mathsf{p}$ or $\mathsf{q}$
  requires at most $\rsize{S}$ many steps. 
\end{itemize}

When the machine exists the above procedure, the current graph is the $\flds{p}$-minimal graph $S_\ell$. 
The current graph is then written on the output tape in two stages. 
During the first stage, the current graph is traversed from top to bottom, 
and the list of non-deleted nodes $u_1, u_2, \dots$ is written on a separate working tape
in time $\bigO(\rsize{S})$.
Let $s$ be the isomorphism $s(u_i) = i$.
In the second stage, the current graph $S_\ell$ is traversed from top to bottom a second time. 
For each node specification $\langle u_i, \Lab(u_i), \suc(u_i) \rangle$, 
the node specification $\langle s(u_i), \Lab(u_i), s*(\suc(u_i)) \rangle$ is written on the output 
tape. Using a counter and the list of marked nodes $u_1, u_2, \dots$, 
this is achieved in time $\bigO(\rsize{S}^2)$.
The machine outputs an increasing list of node specifications, 
the represented graph is isomorphic to $S_{\ell}$.

We now investigate on the computational complexity of the above procedure. 
All preparatory steps, that is, initialising the current graph, computing the 
node corresponding to $p$
and marking the subterm $\subgraphAt{S}{p}$, require $\bigO(\size{S}^2)$ many steps in total.
Since every time when the machine enters state $\mathsf{p}$ one unmarked node is marked, we 
conclude that the machine enters state $\mathsf{p}$ at most $\size{S} \leqslant \rsize{S}$ often. 
The machine enters state $\mathsf{q}$ either after leaving state $\mathsf{p}$ or when
$S_{i,j} \flda{v_{i,j}}{u_i} S_{i,{j+1}}$ in the reduction \eqref{l:comp:flds:I} holds.
By the previous observation, and employing Lemma~\ref{l:fold:bound} on the sequence 
\eqref{l:comp:flds:I} we see that the constructed TM enters 
state $\mathsf{q}$ at most $\bigO(\size{S}) = \bigO(\rsize{S})$ often.
Since each state transition requires at most $\bigO(\rsize{S})$ many steps, 
we conclude that $S_\ell$ is constructed in time $\bigO(\rsize{S}^2)$.
Finally, writing the normalised representation of $S_\ell$ on the output tape
requires again at most $\bigO(\rsize{S}^2)$ many steps. 
Summing up, the machine operates in time $\bigO(\rsize{S}^2)$. 
This concludes the lemma. 
\end{proof}

\begin{lemma}\label{l:comp:step}
  Let $S$ be a term graph, let $p$ be a position of $S$
  and let $L \to R$ be a rewrite rule of the simulating graph rewrite system.
  It is decidable in time $\bigO(\rsize{S}^2 * 2^{\rsize{L \to R}})$
  whether $S \REW[p, L \to R] T$ for some term graph $T$.
  Moreover, the term graph $T$ 
  is computable from $S$, $p$ and $L \to R$ in time $\bigO(\rsize{S}^2 * 2^{\rsize{L \to R}})$.
\end{lemma}
\begin{proof}
We construct a TM that on input $S$, $p$ and $L \to R$ computes the
reduct $T$ for $S \REW[p, L \to R] T$. If the latter does not hold, the machine rejects.
For this we suppose that the nodes of $L \to R$ are chosen in such a way that
$\nodes[L \to R] = \set{1, \dots, \card{L \to R}}$ (we keep this invariant when constructing the final algorithm).
Let $u$ be the node corresponding to $p$ in $S$.
In \cite[Lemma 24]{AM10} it is shown that there exists a TM operating in time $2^{\bigO(\rsize{L})} * \bigO(\rsize{S}^2)$ that, 
on input $L$, $S$ and $u$, 
either writes on its output-tape the graph morphism $m$ such that $m \colon L \to \subgraphAt{S}{u}$ if it exists, 
or fails. The morphism $m$ is encoded as an associative list, more precisely, a list of 
pairs $(u,m(u))$ for $u \in L$. 
The size of this list is bound by $\bigO(\size{L} * \log(\size{L} + \size{S}))$.
First, this machine is used to compute $m \colon L \to \subgraphAt{S}{u}$, the resulting 
morphism is stored on a working tape. 
For this, the node $u$ is computed in time $\rsize{S}^2$ beforehand.
If constructing the morphism fails, then rule $L \to R$ is not applicable at position $p$, i.e., 
$u$ is not a redex in $S$ with respect to $L \to R$.
The constructed machine rejects. 
Otherwise, the reduct $T$ is computed as follows.

Set $L' \defsym r(L)$ and $R' \defsym r(R)$ for the graph morphism
$r$ defined by $r(v) \defsym v + \size{S}$. Then $R' \cap S = \varnothing$.
We compute $T = \replaceAt{S}{u}{m'(R')}$ for $m' \colon L' \to \subgraphAt{S}{u}$
using the morphism $m \colon L \to \subgraphAt{S}{u}$ as computed above.
Let $f(v) = m(v)$ if $\lab[R](v) \in \VS$ and $f(v) = v + \size{S}$ otherwise.
Since $\nodes[L] \cap \nodes[R] = \Var{R}$ we see that $T = \replaceAt{S}{u}{f(R)}$.

Next, the machine constructs $S \cup f(R)$ on an additional working tape as follows.
First, $S$ is copied on this tape in time linear in $\rsize{S}$.
Simultaneously, $\size{S}$ is computed on an additional tape. 
Using the counter $\size{S}$, $v + \size{S}$ is computable in time $\bigO(\log(\size{R}) + \log(\size{S}))$, 
whereas $m(v)$ for $v \in L \cap R = \Var{R}$ is computable in time $\bigO(\size{L} * \log(\size{L} + \size{R}))$
(traversing the associative list representing $m$).
We bind the complexity of $f$ by $\bigO(\rsize{L} * \rsize{R} * \rsize{S})$ independent on $v$.
Finally, for each node-specification 
$\langle v, \Lab(v), \suc(v) \rangle$ with $\Lab(v) \in \FS$ encountered in $R$, 
the machine appends $\langle f(v), \Lab(v), f^*(\suc(v)) \rangle$.
Employing $L \cap R = \Var(R)$, one verifies that $S \cup f(R)$ is obtained this way.
Overall, the runtime is $\bigO(\rsize{L} * \rsize{R}^2 * \rsize{S})$.

Employing $\rsize{f(R)} = \bigO(\size{R} * \log(\max(\size{R},\size{S})))$, 
we see that $S \cup f(R)$ can be bound in size by $\bigO(\rsize{S} * \rsize{R})$.
To obtain $T = \subgraphAt{(S \cup m(R'))}{v} = \subgraphAt{(S \cup f(R))}{v}$ 
for $v$ either $\grt(m(R'))$ or $\grt(S)$, 
the quadratic marking algorithm of Lemma~\ref{l:comp:marking} is used.
Finally, the marked subgraph obeying the standard encoding is written onto the output tape as in Lemma~\ref{l:comp:flds}.

We sum up: it takes at most $2^{\bigO(\rsize{L})} * \bigO(\rsize{S}^2)$ many steps to compute the morphism $m$. 
The graph $S \cup m(R')$ is obtained in time $\bigO(\rsize{L} * \rsize{R}^2 * \rsize{S})$.
Marking $T$ in $S \cup m(R')$ requires at most $\bigO(\rsize{S \cup m(R')}^2) = \bigO(\rsize{S}^2 * \rsize{R}^2)$
many steps.
Finally, the reduct $T$ is written in $\rsize{S \cup f(R)}^2 = \bigO(\rsize{S}^2 * \rsize{R}^2)$
steps onto the output-tape.
Overall, the runtime is bound by $2^{\bigO(\rsize{L})} * \bigO(\rsize{S}^2 * \rsize{R}^2)$ worst case.
\end{proof}

\begin{lemma}\label{l:comp:step:complete}
  Let $S$ be a term graph and let $\GS(\RS)$ be the simulating graph rewrite 
  system of $\RS$. 
  If $S$ is not a normal-form of $\GS(\RS)$
  then there exists a position $p$ and rule $(L \to R) \in \GS(\RS)$ 
  such that a term graph $T$ with $S \REWS[\GS(\RS), p, L \to R] T$ 
  is computable in time $\bigO(\rsize{S}^3)$.
\end{lemma}
\begin{proof}
The TM searches for a rule $(L \to R) \in \GS$ and position
$p$ such that $S \REWS[\GS(\RS),p, L \to R] T$ for some term graph $T$.
For this, it enumerates the rules $(L \to R) \in \GS$ on a separate 
working tape. 
For each rule $L \to R$, each node $u \in S$ and some $p \in \Pos[S]$ 
it computes $S_1$ such that $S \flds{p}^! S_1$ in time quadratic 
in $\rsize{S}^2$ (c.f. Lemma~\ref{l:comp:flds}).
Using the machine of Lemma~\ref{l:comp:step}, it decides in time 
$2^{\bigO(\rsize{L})} * \bigO(\rsize{S_1}^2)$ whether rule $L \to R$ applies to 
$S_1$ at position $p$. Since $\RS$ is fixed, $2^{\bigO(\rsize{L})}$ is constant, 
thus the TM decides whether 
rule $L \to R$ applies in time $\bigO(\rsize{S_1}^2) = \bigO(\rsize{S}^2)$.
Note that the choice of $p \in \Pos[S](u)$ is irrelevant, since 
$S \flds{p_i}^! S_1$ and $S \flds{p_j}^! S_2$ for $p_i,p_j \in \Pos[S](u)$ 
implies $S_1 \isomorphic S_2$. Hence the node corresponding to $p_i$ in $S_1$
is a redex with respect 
to $L \to R$ if and only if the node corresponding to $p_j$ is.
Suppose rule $L \to R$ applies at $\subgraphAt{S_1}{p}$.
One verifies $\subgraphAt{S_1}{p} \isomorphic \subgraphAt{S_2}{p}$
for term graph $S_2$ such that $S \ufldp{p}^! \cdot \flds{p}^! S_2$.
We conclude $S \REWS[\GS(\RS),p, L \to R] T$ for some position $p$ 
and rule $(L \to R) \in \GS(\RS)$ if and only if the above procedure succeeds. 
From $u$ one can extract some position $p \in \Pos[S](u)$ in time quadratic in $\rsize{S}$.
This can be done for instance by
implementing the function $\mathsf{pos}(u) = \varepsilon$ if $u = \grt(S)$
and $\mathsf{pos}(u) = p i$ for some node $v \in S$ with $v \suci[S]{i} u$ and $\mathsf{pos}(v) = p$.
Overall, the position $p \in \Pos[S]$ and rule $(L \to R) \in \GS$ is found 
if and only if $S \REWS[p, L \to R] T$ for some term graph $T$.
Since $\size{S} \leqslant \rsize{S}$ nodes, and only a constant 
number of rules have to be checked,
the overall runtime is $\bigO(\rsize{S}^3)$.

To obtain $T$ from $S$, $p$, and $L \to R$, the machine now combines 
the machines from Lemma~\ref{l:comp:ufldp}, Lemma~\ref{l:comp:flds}
and Lemma~\ref{l:comp:step}. These steps can even be performed in time $\bigO(\rsize{S}^2)$, 
employing that the size of intermediate graphs is bound linear in the size of $S$ 
(compare Lemma~\ref{l:fold:bound})
and that sizes of $(L \to R) \in \GS(\RS)$ are constant.
\end{proof}

\begin{lemma}\label{l:comp:sequence}
  Let $S$ be a term graph and let $\ell \defsym \dl(S, \REWS[\GS(\RS)])$. Suppose $\ell = \Omega(\size{S})$.
  \begin{enumerate}
  \item Some normal-form of $S$ that is computable in deterministic time $\bigO(\log(\ell)^3 * \ell^7 )$.
  \item Any normal-form of $S$ is computable in nondeterministic time $\bigO(\log(\ell)^2 * \ell^5 )$.
  \end{enumerate}
\end{lemma}
\begin{proof}
We prove the first assertion.
Consider the normalising derivation
\begin{equation}
  \label{eq:t:derivation}
  \tag{$\dag$}
  S = T_0 \REWS[\GS(\RS)] \dots \REWS[\GS(\RS)] T_l = T \tpkt
\end{equation}
where, for $0 \leqslant i < l$, $T_i$ is obtained from $T_{i+1}$ as given by Lemma~\ref{l:comp:step:complete}.
By Lemma~\ref{l:space}, we see $\size{T_i} \leqslant (\ell + 1) \size{S} + \ell^2 \Delta = \bigO(\ell^2)$.
Here the latter equality follows by the assumption $\ell = \Omega(\size{S})$.
Recall $\rsize{T_i} = \bigO(\log(\size{T_i}) * \size{T_i})$ ($0 \leqslant i < l$) and 
hence $\rsize{T_i} = \bigO(\log(\ell^2) * \ell^2) = \bigO(\log(\ell) * \ell^2)$.
From this, and Lemma~\ref{l:comp:step:complete}, we obtain that
$T_{i+1}$ is computable from $T_i$ in time $\bigO(\rsize{T_i}^3) = \bigO(\log(\ell)^3 * \ell^6)$.
Since $l \leqslant \dl(S, \REWS[\GS(\RS)]) = \ell$ we 
conclude the first assertion.

We now consider the second assertion. 
Reconsider the proof Lemma~\ref{l:comp:step:complete}. For a given rewrite-position $p$, 
a step $S \REWS[\GS(\RS)] T$ can be performed in time $\bigO(\rsize{S})$.
A nondeterministic TM can guess some position $p$, and verify whether 
the node corresponding to $p$ is a redex in time $\bigO(\rsize{S^2})$.
In total, the reduct $T$ can be obtained in nondeterministic time $\bigO(\rsize{S^2})$.
Hence, following the proof of the first assertion, one 
easily verifies the second assertion.
\end{proof}


\section{Discussion}\label{Discussion}

We present an application of our result
in the context of \emph{implicit computational complexity theory}
(see also~\cite{LM09,LM09b}).

\begin{definition}
\label{d:computation}
Let $\NA \subseteq \Val$ be a finite set of \emph{non-accepting patterns}.
We call a term $t$ \emph{accepting} (with respect to $\NA$) if there exists 
no $p \in \NA$ such that $p\sigma = t$ for some substitution $\sigma$.
We say that $\RS$ \emph{computes the relation $R \subseteq {\Val \times \Val}$} 
with respect to $\NA$ if there exists $\m{f} \in \DS$ such that 
for all $s,t \in \Val$,
\begin{equation*}
{R(s,t)} \defiff {\m{f}(s) \rsn[\RS] t} \text{ and $t$ is accepting}\tpkt
\end{equation*}
On the other hand, we say that a relation $R$ is computed by $\RS$ 
if $R$ is defined by the above equations with respect to \emph{some} set $\NA$ of 
non-accepting patterns.
\end{definition}
For the case that $\RS$ is \emph{confluent}
we also say that $\RS$ computes the (partial) \emph{function} 
induced by the relation $R$.

The reader may wonder why we restrict to binary relations, but this
is only a non-essential simplification that eases the presentation.
The assertion that for normal-forms $t$, $t$ is accepting
amounts to our notion of \emph{accepting run} of a TRS $\RS$. This
aims to eliminate by-products of the computation that should not be
considered as part of the relation $R$. (A typical example would be the
constant $\perp$ if the TRS contains a rule of the form $l \to \perp$ and
$\perp$ is interpreted as \emph{undefined}.) The restriction that
$\NA$ is finite is essential for the simulation results below: If we implement
the computation of $\RS$ on a TM, then we also have to be able to
effectively test whether $t$ is accepting.

To compute a relation defined by $\RS$, 
we encode terms as graphs and perform graph rewriting
using the simulating GRS $\GS(\RS)$. 

\begin{theorem} \label{t:precise}
  Let $\RS$ be a terminating TRS, moreover suppose $\rcR(n) = \bigO(n^k)$ for all $n \in \NAT$ and some $k \in \NAT$, $k\geqslant 1$.
  The relations computed by $\RS$ are computable in nondeterministic time $\bigO(n^{5k+2})$.
  Further, if $\RS$ is confluent then the functions computed by $\RS$ are computable in deterministic time $\bigO(n^{7k+3})$. 
  Here $n$ refers to the size of the input term.
\end{theorem}
\begin{proof}
We investigate into the complexity of a relation $R$ computed by $\RS$.
For that, single out the corresponding defined function symbol $\m{f}$ and fix 
some argument $s \in \Val$.
Suppose the underlying set of non-accepting patterns is $\NA$.
By definition, $R(s,t)$ if and only if $\m{f}(s) \rsn t$ and $t \in \Val$ is accepting with respect to $\NA$.
Let $S$ be a term graph such that $\trepr(S) = \m{f}(s)$ and recall that
$\size{S} \leqslant \size{\m{f}(s)}$.
Set $\ell \defsym \dl(S, \REWS[\GS(\RS)])$.
By the Adequacy Theorem \ref{t:adequacy}, we conclude 
$S \REWSL[\GS(\RS)]{!} T$ where $\trepr(T) = t$,
and moreover, $\ell \leqslant \rcR(\size{\m{f}(s)}) = \bigO(n^k)$. 
By Lemma~\ref{l:comp:sequence} we see that $T$ is computable from $S$ in nondeterministic time
$\bigO(\log(\ell)^2 * \ell^5 ) = \bigO(\log(n^k)^2 * n^{5k}) = \bigO(n^{5k+2})$.
Clearly, we can decide in time linear in $\rsize{T} = \bigO(\ell^2) = \bigO(n^{2k})$ (c.f. Lemma~\ref{l:space})
whether $\trepr(T) \in \Val$, further in time quadratic in $\rsize{T}$ whether $\trepr(T)$ is accepting. 
For the latter, we use the matching algorithm of Lemma~\ref{l:comp:step} on the 
fixed set of non-accepting patterns, where we employ $p\sigma = \trepr(T)$ if and only if 
there exists a morphism $m \colon P \to T$ for $P \in \Tree(p)$ (c.f. Lemma~\ref{l:match:lhs} and Lemma~\ref{l:subst:l}).
Hence overall, the accepting condition can be checked in (even deterministic) time $\bigO(n^{4k})$.
If the accepting condition fails, the TM rejects, otherwise it accepts a term graph $T$ representing $t$.
The machine does so in nondeterministic time $\bigO(n^{5k+2})$ in total. 
As $s$ was chosen arbitrary, we conclude the first half of the theorem.

Finally, the second half follows by identical reasoning, where we use the 
deterministic TM as given by \ref{l:comp:sequence} instead of the nondeterministic one.
\end{proof}

Let $R$ be a binary relation such that
$R(x,y)$ can be decided by some \emph{nondeterministic} TM in 
time polynomial in the size of $x$.
The \emph{function problems} $R_F$ associated with $R$ is:
given $x$, find \emph{some} $y$ such that $R(x,y)$ holds. 
The class $\FNP$ is the class of all functional problems defined in the above way, 
compare \cite{Papa}. 
$\FP$ is the subclass resulting if we only consider function problems in $\FNP$ that 
can be solved in polynomial time by some deterministic TM.
As by-product of Theorem \ref{t:precise} we obtain:
\begin{corollary}
  Let $\RS$ be a terminating TRS with polynomially bounded runtime complexity.
  Suppose $\RS$ computes the relation $R$.
  Then $R_F \in \FNP$ for the function problem $R_F$ associated with $R$.
  Moreover, if $\RS$ is confluent then $R_F \in \FP$.
\end{corollary}
\begin{proof}
The nondeterministic TM $M$ as given 
by Theorem~\ref{t:precise} (the deterministic TM $M$, respectively) 
can be used to decide membership $(s,t) \in R$.
Observe that by the assumptions on $\RS$, 
the runtime of $M$ is bounded polynomially in the size of $s$. Recall that 
$s$ is represented as some term graph $S$ with $\trepr(S) = s$, in particular 
$s$ is encoded over the alphabet of $M$ in size
$\rsize{S} = \bigO(\log(\size{S}) * \size{S})$ for $\size{S} \leqslant \size{s}$.
Thus trivially $M$ operates in time polynomially in the size of $S$.
\end{proof}

\bibliographystyle{plain} 

\end{document}